\documentclass[11pt]{article}
\usepackage[utf8]{inputenc}
\usepackage{multicol}
\usepackage{amsfonts, amsmath, amssymb, amsthm,textcomp}
\usepackage{array, makecell}
\usepackage{pgf, tikz, color, xcolor, soul}
\usepackage{pgfplots}
\usepackage{bbold}
\usepackage{bm}
\usepackage{enumitem}
\usepackage{caption}
\usepackage{subcaption}

\usepackage{verbatim}

\usepackage{xcolor}
\tikzset{%
	dots/.style args={#1per #2}{%
		line cap=round,
		dash pattern=on 0 off #2/#1
	}
}
\definecolor{light-gray}{gray}{0.7}
\definecolor{dark-gray}{gray}{0.3}

\usetikzlibrary{spy}
\usetikzlibrary{backgrounds}
\usetikzlibrary{decorations}
\usetikzlibrary{patterns}
\usetikzlibrary{arrows,matrix,positioning,calc}

\newtheorem*{theorem*}{Theorem}
\newtheorem{theorem}{Theorem}
\newtheorem{proposition}{Proposition}

\newtheorem{lemma}[theorem]{Lemma}

\newtheorem{example}{Example}

\newcommand{\remove}[1]{}

\theoremstyle{definition}

\theoremstyle{remark}
\newtheorem{remark}{Remark}

\usepackage{xcolor}
\newif\ifcomment
\commenttrue

\newcommand{\C}{\mathcal{C}}
\newcommand{\x}{\boldsymbol{x}}
\renewcommand{\c}{\boldsymbol{c}}
\newcommand{\e}{\boldsymbol{e}}

\newcommand{\z}{\boldsymbol{z}}

\newcommand{\y}{\boldsymbol{y}}

\renewcommand{\epsilon}{\varepsilon}
\definecolor{darkgray}{RGB}{64,64,64}
\definecolor{litegray}{RGB}{192,192,192}
\tikzstyle{block}=[draw, rectangle, minimum height=1cm, text width=1.5cm, text centered, draw=darkgray, font=\small]
\tikzstyle{block_medium}=[draw, rectangle, minimum height=1.5cm, text width=2cm, text centered, draw=darkgray, font=\small]
\tikzstyle{block_large}=[draw, rectangle, minimum height=1.75cm, text width=3cm, text centered, draw=darkgray, font=\small]
\tikzstyle{line} = [draw, -latex]
\title{Feedback Insertion-Deletion Codes}
\author{Georg Maringer \thanks{Technical University of Munich. Supported by the German Research Foundation (Deutsche Forschungsgemeinschaft, DFG) under Grant No.~WA3907/4-1. Email: {\tt georg.maringer@tum.de}} \and Nikita Polyanskii\thanks{Technical University of Munich and Skolkovo Institute of Science and Technology.  Supported in part  by the German Research Foundation (Deutsche Forschungsgemeinschaft, DFG) under Grant No. WA3907/1-1. Email: {\tt nikita.polyansky@gmail.com}.} \and Ilya Vorobyev \thanks{Skolkovo Institute of Science and Technology. 
Supported in part by RFBR and JSPS under Grant No.~20-51-50007, and by RFBR under Grant No.~20-01-00559. Email: {\tt vorobyev.i.v@yandex.ru}.} \and Lorenz Welter \thanks{Technical University of Munich. Supported by the European Research Council (ERC) under the European Union’s Horizon 2020 research and innovation programme (grant agreement No. 801434). Email: {\tt lorenz.welter@tum.de}}}

\date{}

\pgfplotsset{compat=1.15}

\begin{document}
\maketitle
\begin{abstract}
  In this paper, a new problem of transmitting information over the adversarial insertion-deletion channel with feedback is introduced. 
  Suppose that the encoder transmits $n$ binary symbols one-by-one over a channel, in which some symbols can be deleted and some additional symbols can be inserted. 
  After each transmission, the encoder is notified about the insertions or deletions that have occurred within the previous transmission and the encoding strategy can be adapted accordingly. 
  The goal is to design an encoder that is able to transmit error-free as much information as possible under the assumption that the total number of deletions and insertions is limited by $\tau n$, $0<\tau<1$. We show how this problem can be reduced to the problem of transmitting messages over the substitution channel. Thereby, the maximal asymptotic rate of feedback insertion-deletion codes is completely established. The maximal asymptotic rate for the adversarial substitution channel has been partially determined by Berlekamp and later finished by Zigangirov.
  However, the analysis of the lower bound by Zigangirov is quite complicated.
  We revisit Zigangirov's result and present a more elaborate version of his proof.
\end{abstract}
\newpage
\section{Introduction}
The problem of constructing codes capable of correcting insertion and deletion errors has been studied since 1965 when Levenshtein published his work~\cite{Levenshtein1966binary}, in which he established that any code is capable of correcting $t$ deletions if and only if it is capable of correcting $t$ insertions and deletions. Furthermore, he showed that the Varshamov-Tenengolts codes~\cite{varshamov1965code} can be used to correct a single insertion or deletion error. Notably, most of the classic techniques for constructing error-correcting codes cannot be applied to construct codes coping with insertions and deletions because these types of errors induce a loss of synchronization between the sender and the receiver. This resulted in substantial development of new approaches and methods for the deletion channel in recent years~\cite{cheraghchi2020overview}.

We first recap the results for adversarial deletion channels without the presence of noiseless feedback to the sender. Let a code $\C\subseteq \{0,1\}^n$ consist of $2^{Rn}$ binary codewords of length $n$. The goal is to find a code maximizing the achievable \textit{rate} $R$ which is able to correct up to a fraction of $\tau$ \textit{adversarial} deletions. 
By adversarial deletions we refer to the case when the adversary can inflict an arbitrary pattern of deletions with full knowledge of the code and the codeword to be transmitted.
Recall that the error correction capability of $\C$ can be defined using the notion of the longest common subsequence between two codewords. For a code being able to correct a fraction of $\tau$ errors the longest common subsequence between any two codewords in $\C$ is of length less than $(1-\tau)n$.
It is clear that for $\tau \ge \frac{1}{2}$, the adversary can force the channel to output either the all-one word or the all-zero word. Thus, a code $\C$ can contain at most two codewords and the asymptotic rate is zero. The first code construction capable of correcting a non-zero fraction of deletions was proposed by Schulman and Zuckerman~\cite{schulman1999asymptotically}. They proposed to use concatenated codes composed of non-binary outer codes and well-performing short binary inner codes that can be found by brute-force search. The last improvement in this direction is due to Bukh, Guruswami, and Hastad~\cite{bukh2016improved2,bukh2016improved}.
They provided a family of code constructions with a positive rate for any $\tau<\sqrt{2}-1\approx 0.41$. 
For the adversarial 
substitution channel it is well known~\cite{plotkin1960binary} that it is not possible to have exponentially many codewords in a code tolerating a fraction of $\tau\ge 0.25$ errors.  
However, to the best of our knowledge not much further progress has been made on the limitations of deletion-correcting codes. Even determining the maximal fraction of deletions for which there exist codes with rate bounded away from zero remains an open research problem.

 \textit{Noiseless feedback} between sender and receiver can potentially increase the maximal rate of a code with fractional error correction capability $\tau$. 
A feedback model with adversarial substitutions was investigated in the paper~\cite{berlekamp1964block} by Berlekamp. 
In this setting a binary channel can flip at most a fraction of $\tau$ symbols within a block. Additionally after each symbol transmission the sender gets notification about which symbol has been received by utilizing a noiseless feedback channel. Therefore, before sending the next symbol the encoder is able to adjust the encoding strategy according to the previously received symbols.
The maximal asymptotic rate of a feedback code for this channel has been completely characterized by~Berlekamp~\cite{berlekamp1964block} and Zigangirov~\cite{zigangirov1976number}. Interestingly, their results show that the maximal asymptotic rate is positive for $\tau<1/3$. 
\subsection{Notation}

This section formally defines notations that are used throughout this paper. The set of integers $\{ 1,\ldots , n \}$ is denoted by $[n]$. The set of integers $\{i+1,i+2,\ldots, j\}$ is denoted by $(i,j]$. We refer to the binary alphabet as $\{0,1\}$ and to a binary string of length $n$ as $\x \in \{0,1\}^n$, i.e. $\bm{x} = (x_1,\ldots,x_n)$. We use $x_i \in \{0,1\}$ to denote the $i$th element of the string $\bm{x}$ where $i \in [n]$. The set $\{0,1\}^{*}$ contains all binary strings of variable length including the empty string which is denoted as $(\;)$. We refer to the length of $\bm{x}$ as $\vert \bm{x} \vert$. Given two strings $\bm{x} \in \{0,1\}^{n_1}$ and $\bm{y} \in \{0,1\}^{n_2}$, we denote by $\bm{z} = ( \bm{x} || \bm{y} )$  the concatenation of the strings, hence $\bm{z} \in \{0,1\}^{n_1+n_2}$. It is also possible to write a binary string $\bm{x}=(x_1,\ldots,x_n) \in \{0,1\}^n$ as a concatenation of its components, i.e. $\bm{x} = (x_1||\ldots||x_n)$. However, for $n$ strings $\y_1,\ldots,\y_n\in \{0,1\}^*$, we distinguish the concatenation of strings $\y = (\y_1||\ldots||\y_n)\in \{0,1\}^*$ from a tuple $\hat\y = (\y_1,\ldots,\y_n) \in \left(\{0,1\}^*\right)^n$ containing several strings. The tuple $\hat\y$ can be uniquely split back into $\y_1,\ldots,\y_n$, whereas for the concatenated string $\y$ this is not possible.

We use the following two distance notations. For any two strings $\bm{x} \in \{0,1\}^{*}$ and $\bm{y} \in \{0,1\}^{*}$ we denote as $\Delta(\bm{x},\bm{y})$ the minimal number of deletions and insertions required to obtain $\bm{x}$ from $\bm{y}$. This quantity is frequently referred to as the \textit{longest common subsequence} distance between strings $\bm{x}$ and $\bm{y}$. Additionally, we use the notation $d_H(\bm{x},\bm{y})$ to denote the Hamming distance between strings $\bm{x}$ and $\bm{y}$ of same length. All logarithms are base $2$ unless otherwise indicated. The binary entropy function is defined as $h(x):= -x \log x - (1-x) \log(1-x)$.
\subsection{Problem statement}
In this paper the problem of communicating over the adversarial insertion-deletion channel with feedback is addressed. 
Practically this problem is relevant because certain channels, e.g. DNA-based storage channel, are prone to inflict insertion-deletion errors rather than substitution errors.
For any message $m\in[M]$, the sender's goal is to encode $m$ into a binary string $\c\in\{0,1\}^n$ such that after transmitting this string over the binary adversarial insertion-deletion channel, the receiver is able to correctly decode the message $m$ from the received string $\y$. The output string $\y$ is controlled by an adversary (channel noise) who can change the channel input string $\bm{c}$ by inflicting insertions and deletions. 
 The adversary succeeds if the decoder's decision is different from the message $m$. The process of encoding and transmitting consists of $n$ steps, i.e. $\c=(c_1,\ldots,c_n)$. At the $i$th moment the encoder generates a binary symbol $c_i\in\{0,1\}$ and transmits it over the channel. The adversary takes $c_i$ and is able to inflict insertion and deletion errors to create the output $\y_i\in\{0,1\}^{*}$. We stress that the adversary has the ability to place the insertions before the symbol $c_i$ and not only after the transmitted symbol $c_i$, increasing his flexibility. 
 The adversary has full knowledge of the message $m$ as well as the encoding and decoding strategies which encoder and decoder deploy, while the decoder only gets access to the entire concatenated output ${\y} = (\y_1||\ldots||\y_n)$ but not its partitioning, i.e. the decoder is not able to split it into the respective $\y_i$. This means that the decoding function $dec({\y})$ is of the form $dec: \{0,1\}^* \to [M]$. At the $(i+1)$st moment, the sender can adapt the further encoding strategy for the message $m$ based on the tuple of the strings $\hat\y^{(i)}:=(\y_1,\ldots,\y_i)\in\left(\{0,1\}^*\right)^i$. Thereby, $c_{i+1}: [M]\times \left(\{0,1\}^*\right)^i\to\{0,1\}$ is a function of $m$ and $\hat \y^{(i)}$, i.e.,  $c_{i+1}=c_{i+1}(m,\hat \y^{(i)})$. Notice that at this point the adversary is not allowed to insert symbols into $\y_i$ anymore. We require the total number of errors that the adversary can create to be at most $t$, that is
\begin{equation}\label{eq::sum distances}
\sum_{i=1}^{n}\Delta(c_i(m,\hat \y^{(i-1)}), \y_i)\le t.
\end{equation}
Our goal is to find the maximum rate at which the sender can transmit a message in such a way that based on the received string $\y$, the receiver is always able to decode the message without error, given the maximal number of errors $t$ the adversary may induce. Formally, let $M_{id}(n,t)$ be the maximum number of messages the sender can transmit to the receiver under the conditions imposed by this model. We discuss the case that the maximal number of errors is proportional to $n$, i.e., $t=\lfloor \tau n\rfloor$ with $0\le \tau \le 1$ and define the \textit{maximal asymptotic rate} of feedback codes capable of correcting a fraction $\tau$ of insertion-deletion errors to be
$$
R_{id}(\tau):=\limsup_{n\to\infty} \frac{\log M_{id}(n,\lfloor \tau n \rfloor)}{n}.
$$
Notice that the channel output for the insertion-deletion channel could consist of up to $n+\lfloor \tau n\rfloor$ symbols when the channel is used $n$ times. However, the asymptotic rate is measured in bits per channel use.
The main goal of this paper is to find the quantity $R_{id}(\tau)$ for arbitrary $\tau \in [0,1]$.
\subsection{Our contribution}
We show that the problem of transmitting information over the adversarial insertion-deletion channel can be reduced to the problem of transmitting information over the adversarial substitution channel. More specifically, we first demonstrate that any encoding algorithm with $n$ transmissions tackling $t$ insertions can serve as an encoding algorithm with $n+t$ transmissions correcting $t$ substitution errors. From that result we obtain an upper bound on $R_{id}(\tau)$. Then we adapt an encoding strategy originally suggested for the feedback substitution channel to obtain a lower bound on $R_{id}(\tau)$. This lower bound matches the previously derived upper bound. 
Therefore, we have determined $R_{id}(\tau)$ for any fraction of errors $\tau$, $0\le \tau \le 1$. The resulting asymptotic rate is plotted in Figure~\ref{fig::asymptotic rate id codes}.
\begin{theorem}\label{th::main result}
	The maximal asymptotic rate of feedback codes for the adversarial insertion-deletion channel is
	\begin{equation}\label{eq::main result}
	R_{id}(\tau)= \begin{cases}
	(1+\tau)\left(1-h\left(\frac{\tau}{1+\tau}\right)\right),\quad &\text{ for }0\le\tau\le\sqrt{5}-2, \\
	(1-2\tau)\log\left(\frac{1+\sqrt{5}}{2}\right),\quad &\text{ for }\sqrt{5}-2< \tau\le \frac{1}{2} ,\\
	0,\quad &\text{ otherwise.}
	\end{cases}
	\end{equation}
\end{theorem}
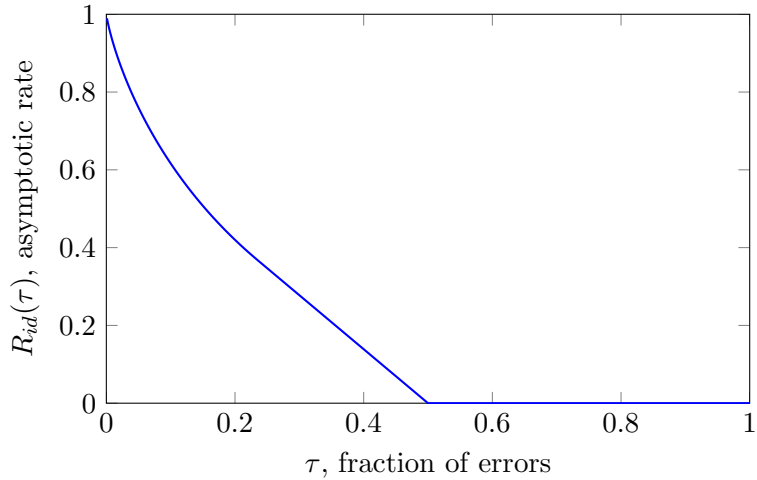
\begin{figure}[t]
\centering
\begin{tikzpicture}
\pgfplotsset{compat = 1.3}
\begin{axis}[
	legend style={nodes={scale=1.0, transform shape}},
	legend cell align={left},
	width = 0.6\columnwidth,
	height = 0.4\columnwidth,
	xlabel = {$\tau$, fraction of errors},
	xlabel style = {nodes={scale=1.0, transform shape}},
	ylabel = {$R_{id}(\tau)$, asymptotic rate},
	ylabel style={nodes={scale=1.0, transform shape}},
	xmin = 0,
	xmax = 1.0,
	ymin = 0.0,
	ymax = 1.0,
	legend pos = south east]
\addplot[color=blue, mark=none,thick] table {plotIDcode.txt};
\end{axis}
\end{tikzpicture}
  \caption{Maximal asymptotic rate of binary feedback insertion-deletion codes.}
  \label{fig::asymptotic rate id codes}
\end{figure}
As a side contribution of our paper, we present a more elaborate version of Zigangirov's technical analysis~\cite{zigangirov1976number} of Horstein's algorithm~\cite{horstein1963sequential} for the adversarial substitution channel. 
We hope that this makes the intuitive algorithm which is nevertheless hard to analyze more accessible to a wider audience.
\subsection{Outline}
The remainder of the paper is organized as follows. In Section~\ref{sec:upper bounds on the rate:}, we provide an upper bound on $R_{id}(\tau)$. In Section~\ref{sec::lower bounds on the rate}, we discuss a feedback insertion-deletion code and its reduction to feedback substitution codes. The analysis of the feedback substitution code is given in Section~\ref{sec::analysis of the feedback substitution code}. 
Finally, we conclude the paper with Section~\ref{sec::conclusion}.

\section{Upper bound on $R_{id}(\tau)$}\label{sec:upper bounds on the rate:}
In this section we introduce the concept of the adversarial substitution channel and prove that a feedback insertion-deletion code correcting a fraction of $\tau$ errors can serve as a feedback substitution code capable of correcting a fraction of $\tau/(1+\tau)$  errors. This enables us to prove an upper bound on the maximal achievable rate for the insertion-deletion channel which is equal to $R_{id}(\tau)$ specified in~\eqref{eq::main result}.
\subsection{Feedback adversarial substitution channel}\label{sec::problem statement for bsc}
First we specify the adversarial \textit{substitution channel} with feedback. In this work we are only considering the binary case, thus the channel is specified to have binary alphabets at its input as well as its output. The channel is synchronized and an error is defined as the event that an output symbol is not equal to its respective input symbol (e.g. input symbol ``0'' is flipped to output symbol ``1'').
Thus the transmission of a single bit can be illustrated as shown in Figure \ref{fig:Bsc}. This channel is also known in the coding theory literature as the \textit{binary symmetric channel} or the \textit{bit-flip channel}.
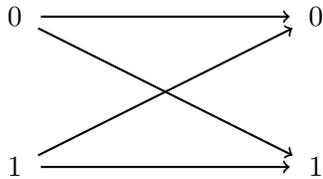
\begin{figure}[h]
\centering
\begin{tikzpicture}
\node (a) [circle] at (0,0) {1};
\node (b) [circle] at (0,2) {0};
\node (c) [circle] at (4,0) {1};
\node (d) [circle] at (4,2) {0};
\draw[thick,->] (a) -- (c);
\draw[thick,->] (a) -- (d);
\draw[thick,->] (b) -- (c);
\draw[thick,->] (b) -- (d);
\end{tikzpicture}
\caption{Substitution channel} \label{fig:Bsc}
\end{figure}

Let $[M]:= \{1, \dots, M\}$ denote the message set. The sender's task is to transmit a message $m\in [M]$ error-free to the receiver. Similar to the adversarial insertion-deletion channel the adversary's task is to prevent correct decoding of the message by inflicting substitution errors using his knowledge about the encoding and decoding algorithms and the message $m$.  We are considering block encoding where we denote the blocklength by $n$. 
We denote the $i$th symbol the encoder sends over the channel by $c_i\in\{0,1\}$. The adversary takes $c_i$ and is able to inflict a substitution error to create $y_i\in\{0,1\}$.
When generating the $i$th channel input symbol $c_i$, the sender can adjust the encoding strategy according to the previously received symbols by the decoder $\hat {\y}^{(i-1)}:=(y_1,\ldots,y_{i-1})$. In other words, $c_i: [M]\times \{0,1\}^{i-1}\to \{0,1\}$ is a function of $m$ and $\hat {\y}^{(i-1)}$, i.e. $c_i(m,{\hat \y}^{(i-1)})$. Notice that $\hat\y^{(i-1)}$ is equivalent to $\y^{(i-1)}:=(y_1||\ldots ||y_{i-1})\in\{0,1\}^{i-1}$ because the substitution channel is synchronized.
Based on the output string $\y:=(y_1||\ldots||y_n)$, the receiver has to correctly decode the message $m$. This means that the decoding function $dec(\y)$ is of the form $dec: \{0,1\}^n\to [M]$.  We require the total number of errors produced by the adversary to be at most $t$, i.e., $d_H(\c,\y)\le t$, where $\c:=(c_1,\ldots,c_n)$. Let $M_s(n,t)$ be the maximal number of messages the sender can transmit to the receiver and define the maximal asymptotic rate of feedback codes capable of correcting a fraction of $\tau$ substitution errors to be 
$$
R_{s}(\tau):=\limsup_{n\to\infty} \frac{\log M_{s}(n,\lfloor \tau n \rfloor)}{n}.
$$
\subsection{Construction of a substitution code from an insertion-deletion code}
To prove that \eqref{eq::main result} is an upper bound on the maximal rate of feedback codes for the adversarial insertion-deletion channel, we combine two ideas.
First, we prove that a feedback insertion-deletion code can be used to tackle substitution errors.

\begin{lemma}\label{lem::ins-del to subst}
Suppose that a feedback insertion-deletion code of length $n$  and size $M$ capable of correcting $t$ errors is given. Then there exists a feedback code of length $n+t$ and size $M$ capable of correcting $t$ substitution errors. 
\end{lemma}
Second, we make use of an upper bound on the rate $R_s(\tau)$ provided by Berlekamp~\cite{berlekamp1964block}.
\begin{theorem}[{Follows from~\cite[Chapter IV]{berlekamp1964block} and~\cite{zigangirov1976number}}]\label{th:: code rate subst}
The maximal asymptotic rate for the adversarial substitution channel satisfies
	\begin{equation*}
	R_{s}(\tau)=\begin{cases}
	1 - h(\tau),\quad &\text{ for }0\le\tau\le (3-\sqrt{5})/4, \\
	(1-3\tau)\log\left(\frac{1+\sqrt{5}}{2}\right),\quad &\text{ for }(3-\sqrt{5})/4< \tau\le 1/3,\\
	0,\quad &\text{otherwise.}
	\end{cases}
	\end{equation*}
\end{theorem}
By Lemma~\ref{lem::ins-del to subst}, we have $M_s(n+t,t)\ge M_{id}(n,t)$ and, thus, $(1+\tau) R_{s}(\tau/(1+\tau))\ge R_{id}(\tau)$. 
This bound and Theorem~\ref{th:: code rate subst} yield that the asymptotic rate specified in Theorem~\ref{th::main result} is indeed an upper bound on the maximal achievable rate of the adversarial insertion-deletion channel. 
So, it remains to check the validity of Lemma~\ref{lem::ins-del to subst}.
\begin{proof}[Proof of Lemma~\ref{lem::ins-del to subst}]
Let the encoding functions of the given feedback insertion-deletion code  be determined by $c_i: [M]\times \left(\{0,1\}^{*}\right)^{i-1}\to \{0,1\}$ for $i\in[n]$ and let the decoding function be given as $dec: \{0,1\}^* \to [M]$. We shall define encoding functions $e_j:[M]\times\{0,1\}^{j-1}\to\{0,1\}$ for all $j\in[n+t]$ for the substitution channel using the encoding functions $c_i$ as components. Using such functions $e_j$ the sender would be able to transmit the message over the channel with at most $t$ substitution errors.

We denote the number of correctly received symbols by $i$, the number of transmitted symbols by $j$, the index of the last correctly received symbol by $k$,  the received output string by $\z$, and tuple of output strings according to the output of a potential insertion-deletion channel by $\hat\y$.

The encoding scheme we propose consists of three parts. The Initialisation step is for initializing counters and state variables.
Step 1 describes the process of re-transmitting symbols $c_i$ until they have been received correctly. Step 2 is just used to fill potentially unused symbols within the block with zero symbols.
The state variables $i$ and $k$ are only updated in Step 1.

\textbf{Initialisation:} Let $i\gets 0$, $j\gets 0$, $k\gets 0$, $\z\gets(\;)$, and $\hat \y\gets(\;)$.

\textbf{Step 1:}  Define and transmit $e_{j+1}(m,\z)=c_{i+1}(m, \hat \y)$. Update $j\gets j+1$. Let $z_j$ be the received symbol. Update $\z \gets (\z||z_j)$. If $z_j\neq e_j(m,\z)$, then repeat Step 1. Otherwise, update $i\gets i+1$, define $\y_i\gets \z_{[k+1,j]}$, update $\hat \y\gets (\y_1,\ldots,\y_i)$ and $k\gets j$. If $i=n$, then go to Step 2.  Otherwise, repeat Step 1. 

\textbf{Step 2:} If $j=n+t$, update $\y_{n}\gets(\y_n||\z_{[k+1,n+t]})$, $\hat\y\gets(\y_1,\ldots,\y_n)$ and exit algorithm.  Otherwise, define and transmit $e_{j+1}(m,\z)=0$. Update $j\gets j+1$. Let $z_j$ be the received symbol. Update $\z\gets (\z||z_j)$. Repeat Step 2.

The encoding algorithm successfully terminates after $n+t$ transmissions because we will have $n$ correct receptions within a block of $n+t$ symbols. 
The decoding scheme can just be taken from the insertion-deletion code because the output string $\z$ is a possible output of the insertion-deletion channel with at most $t$ errors for encoding functions $c_i$, $i\in [n]$. Indeed, if we define $\hat\y^{(i)}:=(\y_1,\ldots,\y_{i})$, we obtain
$$
\sum_{i=1}^{n}\Delta(c_i(m,\hat\y^{(i)}), \y_i)= \sum_{i=1}^{n}(|\y_i|-1)= \sum_{i=1}^{n}|\y_i| - n =(n+t)-n=t.
$$
Therefore, the decoder of the insertion-deletion code outputs the correct message, i.e., $dec(\z)=m$. This completes the proof.
\end{proof}

\section{Lower bound on $R_{id}(\tau)$}\label{sec::lower bounds on the rate}
In this section, we provide a feedback insertion-deletion code with an asymptotic  rate arbitrarily close to~\eqref{eq::main result}. We adapt an algorithm originally suggested by Horstein~\cite{horstein1963sequential} and further developed by Schalkwijk~\cite{schalkwijk1971class} and Zigangirov~\cite{zigangirov1976number}. Note that this algorithm has already been used for the adversarial substitution channel with feedback. We point out that not every encoding strategy for the adversarial substitution channel can be adapted to be used for the insertion-deletion channel. A key property of Horstein's algorithm is that in case the channel inflicts an error the sender re-transmits the intended symbol until it is received correctly. We elaborate on the importance of this property by considering a general encoding strategy with feedback. We assume that according to the previously received symbols the encoding strategy implies sending $c_i=0$ over the channel and that the adversary inserts symbol $``1"$ before $c_i$ such that $\y_i=10$. The symbol $``1"$  at the beginning of the sequence $\y_i$ can also be interpreted as being created by the adversary as the last symbol of $\y_{i-1}$. For a general encoding strategy it is possible that the encoder would output $c_i=1$ after observing $\y_{i-1}$, thus the adversary would be able to inflict two errors at the price of one for such an encoding strategy. Horstein's procedure prevents this possibility, making it suitable for the proposed insertion-deletions model in this paper.

\subsection{Feedback insertion-deletion code}\label{sec::feedback insertion-deletion code}
Suppose that at most a fraction of $\tau$ errors can occur and the sender would like to transmit $M=2^{Rn}$ messages where $R=R_{id}(\tau)-\epsilon$ and $\epsilon>0$. Based on $\tau,\ \epsilon$ and $n$, the sender and the receiver choose two non-algebraic numbers $\alpha$ and $\beta$ fulfilling $\alpha + \beta = 2$. 
For $\tau<\sqrt{5}-2$, they take $\alpha$ sufficiently close to $2\tau / (1+\tau)$ and for $\sqrt{5}-2\le \tau<\frac{1}{2}$, $\alpha$ is chosen to be around $(3-\sqrt{5})/2$. We discuss how close this parameter has to be taken in Section~\ref{sec::analysis of the feedback substitution code}.
Then the sender divides the $[0,1]$-segment into $M$ subsegments of length $1/M$ and enumerates them from left to right by the elements of $[M]$. The length of these $M$ segments will vary during transmission. If the sender wishes to transmit the message $m\in[M]$, then the $m$th segment will be called the \textit{true segment}. Let $T(i)$ denote the true segment of length $t(i)$ after the $i$th transmission. Denote the union of segments on the left and, respectively, on the right to $T(i)$ by $L(i)$ and $R(i)$. Now we are ready to describe encoding and decoding procedures of the proposed algorithm, written as $\mathfrak{A_{id}}(M,n,\alpha)$.

\begin{remark}
We remark that the non-algebraic numbers are dense within the set of real numbers, meaning that we can approximate any real number arbitrarily close within the non-algebraic numbers.
\end{remark}

\textbf{Encoding procedure:}  At the $i$th moment the sender checks whether the center of the true segment lies in $[0,\frac{1}{2})$. In this case the encoder transmits the symbol $c_i=0$ over the channel, otherwise the encoder transmits $c_i=1$. The sender observes the channel outputs $\y_i$ of length $n_i$ and modifies the length of all segments in the following manner. The sender runs over the symbols of \\ $\y_i=(y_{i,1},\ldots,y_{i,n_i})\in \{0,1\}^{n_i}$ from left to right. If $\y_i$ is the empty string, then the next symbol is transmitted. Otherwise, if $y_{i,j}=0$, then the length of all segments, that entirely lie inside $[0,\frac{1}{2})$, is enlarged by a factor of $\beta$. The length of segments that lie entirely in $[\frac{1}{2},1]$ is shrunk by the factor $\alpha$. There could be a segment that contains the point $\frac{1}{2}$ (not as its endpoint). Let $x$ be the length of this segment. Then $x$ can be represented as $x=x_0+x_1$, where $x_0$ is the length of part of the segment, which lies in $[0,\frac{1}{2})$. Then the length of this segment is updated from $x$ to $\beta x_0 + \alpha x_1$. If $y_{i,j}=1$ is received, the encoder shrinks the segments left to the point $\frac{1}{2}$ by a factor of $\alpha$ while increasing the segments right to it by a factor of $\beta$. Compared to the case $y_{i,j}=0$ the roles of $\alpha$ and $\beta$ are just swapped. During the transmissions the lengths of the segments change according to the algorithm, however, the order of the segments remains unaltered. In total, the sender makes $n_i$ such updates after the $i$th transmission. It can be readily seen that the total length of all segments is always $1$.  The algorithm is designed to increase the length of the true segment such that it contains the point $\frac{1}{2}$ after the last transmission.

The complexity of this encoding algorithm is $\Theta(Mn)$. However, it is clear from the description that to follow the encoding algorithm, it suffices for the encoder to modify the length of only three segments $L(i), T(i)$ and $R(i)$. Thus, the complexity can be reduced to $\Theta(n)$.

\medskip
\textbf{Decoding procedure:}
After receiving a string $\y=(y_1||\ldots||y_{n'})$ of length $n'\in[n-t,n+t]$, the decoder is able to reconstruct the sender's point of view by varying the length of all segments in the same way as the sender. After $n'$ steps, the decoder outputs the message corresponding to the segment containing the point $\frac{1}{2}$. 
The complexity of this decoding algorithm is $\Theta(Mn)$. 
To reduce the decoding complexity the receiver can reverse the decoding procedure in the following way. The decoder runs over the symbols of $\y$ from right to left and tracks the preimage of the point $\frac{1}{2}$.
Formally, let $P(n'):=\frac{1}{2}$.  For $j\in\{0,1,\ldots,n'-1\}$ and $y_j=0$, define
 $$
 P(j-1):=
 \begin{cases}
 \frac{P(j)}{\beta},\quad &\text{if }P(j)\le \frac{\beta}{2},\\
 \frac{1}{2}+\frac{2P(j)-\beta}{2\alpha},\quad &\text{if }P(j)> \frac{\beta}{2}.\\
 \end{cases}
 $$
 For $y_j=1$, define
 $$
 P(j-1):=
 \begin{cases}
 \frac{P(j)}{\alpha},\quad &\text{if }P(j)\le \frac{\alpha}{2},\\
 \frac{1}{2}+\frac{2P(j)-\alpha}{2\beta},\quad &\text{if }P(j)> \frac{\alpha}{2}.\\
 \end{cases}
 $$
Finally, the decoder uses $\hat m:=\lceil M\cdot P(0)\rceil$ as an estimate for the message $m$ that the encoder intended to send. Clearly, the complexity of this simplified algorithm is also $\Theta(n)$.
 
 \subsection{Intuition behind the algorithm}
 We give an intuitive explanation why this strategy is sensible. 
 In this section we are only considering the case that $\tau<\sqrt{5}-2$ as only for this case $$
 R_{id}(\tau)=(1+\tau)\left(1-h\left(\frac{\tau}{1+\tau}\right)\right).
 $$
 Let us take $\alpha = 2\tau/(1+\tau)$.
Assume that during the encoding procedure the number of times  when the true segment contains the point $\frac{1}{2}$ is $o(n)$ and the number of insertion and deletion errors during transmission is $\tau_{ins} n + o(n)$ and $\tau_{del} n + o(n)$, respectively, where $\tau_{ins}+\tau_{del}\le \tau$. Then the length of the true segment after $n$ transmissions, $t(n)$, can be bounded as follows
	\begin{align*} t(n)&\ge  M^{-1} \alpha^{\tau_{ins} n+o(n)} \beta^{(1-\tau_{del})n + o(n)}\\
&=2^{-(R_{id}(\tau)-\epsilon)n + \tau_{ins} n \log(2\tau/(1+\tau)) + (1-\tau_{del})n \log(2/(1+\tau)) + o(n)}\\
&\ge 2^{-(R_{id}(\tau)-\epsilon)n + \tau n \log(2\tau/(1+\tau)) + n \log(2/(1+\tau)) + o(n)}\\
&= 2^{\epsilon n + o(n)}.
	\end{align*}
	where we used the fact that the function $\tau_{ins} \log(2\tau/(1+\tau)) + (1-\tau_{del})\log(2/(1+\tau))$ provided $\tau_{ins}+\tau_{del}\le \tau $ is minimized at $\tau_{ins}=\tau$ and $\tau_{del}=0$. 
Thus, the length of the true segment would be very large at the end and it would definitely contain the point $\frac{1}{2}$. However, our initial assumption is wrong, and the true segment contains the point $\frac{1}{2}$ more than $o(n)$ times.
To actually prove that the strategy is working, we need more sophisticated arguments.

\subsection{Construction of an insertion-deletion code from a substitution code}
The algorithm  $\mathfrak{A_{id}}(M,n,\alpha)$ described in Section~\ref{sec::feedback insertion-deletion code} was originally proposed in a similar form for the substitution channel. We emphasize that this channel provides synchronization between the sender and the receiver, i.e., after transmitting the symbol $c_i$, the channel outputs exactly one symbol $y_i$.  We will refer to the encoding algorithm for the substitution channel with $n$ channel uses that follows the steps described in Section~\ref{sec::feedback insertion-deletion code} as $\mathfrak{A}_{s}(M,n,\alpha)$. A formal definition of this algorithm for the substitution channel will be given in Section~\ref{sec: feedback subsitution code}. In the following we show that it is possible to achieve the asymptotic rate given in equation~\eqref{eq::main result} for the adversarial insertion-deletion channel by using the algorithm $\mathfrak{A_{id}}(M,n,\alpha)$. We first prove the following Lemma.

\begin{lemma}
If for all $n'\in[n-t,n+t]$, the algorithm $\mathfrak{A}_{s}(M,n',\alpha)$ can be used for successful transmissions of any message $m\in [M]$ over the feedback adversarial substitution channel with at most $t':=\lfloor (n'-n+t)/2 \rfloor$ errors, then the algorithm $\mathfrak{A}_{id}(M,n,\alpha)$ can be used for successful transmissions of any message $m\in [M]$ over the feedback adversarial insertion-deletion channel with at most $t$ errors.
\end{lemma}
\begin{proof}
Let $\y=(\y_1||\ldots||\y_n)=(y_1,\ldots,y_{n'})$ be a possible output string of length $n'\in[n-t,n+t]$ if the algorithm $\mathfrak{A}_{id}(M,n,\alpha)$ is used for transmitting a message $m\in[M]$ over the adversarial insertion-deletion channel with at most $t$ errors. Now we show that this string $\y$  is also a possible output if the algorithm $\mathfrak{A}_{s}(M,n',\alpha)$ is used for transmitting the same message $m\in[M]$ over the adversarial substitution channel with at most $t':=\lfloor (n' - n+t) /2 \rfloor$ errors.  Define $\hat\y^{(i-1)}:=(\y_1,\ldots,\y_{i-1})\in\left(\{0,1\}^*\right)^{i-1}$ and $\y^{(i-1)}:=(y_1,\ldots,y_{i-1})\in\{0,1\}^{i-1}$. Let $c_i=c_i(m,\hat \y^{(i-1)})$ and $e_i=e_i(m,\y^{(i-1)})$ be the $i$th bit generated by the encoders of the feedback insertion-deletion code and the substitution code, respectively. We define $\e:=(e_1,\ldots,e_{n'})$. So, it suffices to show that $d_H(\e,\y)\le t'$.  Denote the length of $\y_i$ by $n_i$. We represent $\e$ as the concatenation $(\e_1||\ldots||\e_n)$ such that the binary string $\e_i$ has length $n_i$. It is clear that for ``synchronized" moments $i$ and $N_i:=n_1+\dots+n_{i-1}+1$, both algorithms send the same symbol, i.e., 
$$
c_{i}(m,\hat\y^{(i-1)})=e_{N_i}(m,\y^{(N_i-1)}).
$$
By the definition of the algorithm, it can be readily checked that if $y_j\neq e_j$, then $e_{j+1}=e_j$, i.e., a re-transmission occurs.
Therefore, 
$$
d_H(\e_{i},\y_{i})\le 
\begin{cases}
n_{i} -1, &\text{if }\Delta(c_{i}, \y_{i})= n_{i}-1,\\
n_{i},&\text{if }\Delta(c_{i}, \y_{i})= n_{i}+1.
\end{cases}
$$
Note that the second case can only happen when a deletion occurred. Let $t_{ins}$ be the total number of insertions and $t_{del}$ be the total number of deletions. 
We have the restrictions $t_{ins}+t_{del}\le t$ and $n-t_{del}+t_{ins} = n'$. This implies that
$n' - n + 2 t_{del} \le t$ or $t_{del} \le \lfloor(n-n' + t)/2 \rfloor$. It follows that 
$$
d_H(\e,\y) = \sum_{i=1}^{n}d_H(\e_{i},\y_{i})\le \sum_{i=1}^{n} n_i - \sum_{i=1}^{n}\mathbb{1}\{\Delta(c_{i}, \y_{i})=n_i-1\} \le n' - (n-t_{del})  \le t',
$$
where $\mathbb{1}\{x=a\}$ denotes the indicator function of the event $x=a$.  
\end{proof}

Next we have a closer look at the adversarial substitution channel. In order to achieve the maximal asymptotic rate of the adversarial substitution channel with a fraction $\tau\le (3-\sqrt{5})/4$ of errors, for the parameter $\alpha$ within the algorithm $\mathfrak{A}_{s}(M,n,\alpha)$ it holds that $\alpha \approx 2 \tau$ with $\alpha$ being non-algebraic.
\begin{lemma}[Follows from~\cite{zigangirov1976number}]\label{lem:tangent_hamming}
Let $\alpha^* = 2 \tau^*$ with $\tau^*\le (3-\sqrt{5})/4$ being a non-algebraic number and let the algorithm $\mathfrak{A}_s(M,n,\alpha^*)$ be used to communicate information over the adversarial substitution channel for different $\tau$. Then the achieved rates for different $\tau$ form a tangent line to the point $(\tau^*,R_s(\tau^*))$ of the function $R_s(\tau)$.
\end{lemma}
By using Lemma~\ref{lem:tangent_hamming} we can show the following monotonicity property. 
\begin{proposition}\label{prop::minimization over different}
Let $M_s(n,\tau,\alpha)$ denote the maximum number of messages $M$ that can successfully be transmitted using the algorithm $\mathfrak{A}_s(M,n,\alpha)$ for the adversarial substitution channel with blocklength $n$ and at most $\tau n$ errors. Let $\alpha$ and $\tau$ be non-algebraic such that $\alpha = 2\tau/(1+\tau)$ and $\alpha\le (3-\sqrt{5})/2$. Then it holds that
\begin{align*}
    \min\limits_{k\in\{0,1,\ldots,\tau n\}} \log M_s(\lfloor n(1+\tau)\rfloor-2k,\lfloor\tau n\rfloor-k,\alpha) &= (1+o(1))  \log M_s(\lfloor(1+\tau)n\rfloor,\lfloor \tau n\rfloor,\alpha)\\ &=n(1+\tau) R_s\left(\frac{\tau}{1+\tau}\right)  + o(n).
\end{align*}
\end{proposition}
\begin{proof}
    We fix $\alpha = 2\tau/(1+\tau)$ because it allows us to successfully transmit by the algorithm $\mathfrak{A}_s(M,\lfloor n(1+\tau) \rfloor, \alpha)$ asymptotically  the maximum amount of messages given the number of errors is at most $\lfloor \tau n\rfloor$.  Then our goal is to show that although the chosen $\alpha$ is not optimal in terms of the achievable rate for blocklength $\lfloor n(1+\tau) \rfloor-2k$ and at most $\lfloor \tau n \rfloor - k$ errors, the quantity $\log M_s(\lfloor n(1+\tau)\rfloor-2k,\lfloor\tau n\rfloor-k,\alpha)$ is still asymptotically minimized for $k=0$.
    
    According to Lemma~\ref{lem:tangent_hamming} by fixing $\alpha=2\tau / (1+\tau)$ and using the algorithm $\mathfrak{A}_s(M,n',\alpha)$ for $n' = \lfloor n(1+\tau)\rfloor - 2k$ and at most $\lfloor \tau n \rfloor - k$ errors the achievable rates form the tangent line to the point $(\tau/(1+\tau),R_s(\tau/(1+\tau)))$ of the function $R_s(x)=1-h(x)$.
    
    To prove the required claim concerning the minimization, we define the function
    \begin{equation*}
        f(x): = \frac{\log M_s(n(1+\tau-2x), (\tau-x)n,\alpha)}{n} = (1+\tau-2x) \left(a \frac{\tau-x}{1+\tau-2x}+b\right)(1+o(1)),
    \end{equation*}
    where the real valued $a,b$ are fixed by the requirement that the achieved rates form a tangent to the point $(\tau/(1+\tau),R_s(\tau/(1+\tau)))$.
    A short computation shows that $a = \log\tau$ and $b = 1-\log(1+\tau)$.
    
    The derivative of $f(x)$ as $n\to\infty$ is
    \begin{equation*}
        \frac{\partial f}{\partial x} = -a - 2b + o(1) = \log\left(\frac{(1+\tau)^2}{\tau}\right)  -2 + o(1) > 0,
    \end{equation*}
    where the inequality follows the argument of the logarithm is strictly greater than $4$ for $\tau<1/2$.
    This shows that $f$ is asymptotically minimized for $x=0$, completing the proof.
\end{proof}
We have already shown that the output sequences of the algorithm $\mathfrak{A}_{id}(M,n,\alpha)$ for the adversarial insertion-deletion channel are also within the set of output sequences of the adversarial substitution channel if the algorithm $\mathfrak{A}_s(M,n',\alpha)$ is used for at most $t' = \lfloor (n'-n+t)/2 \rfloor$ errors and $n'\in [n-t,n+t]$. Therefore, from Proposition~\ref{prop::minimization over different} 
it follows that $R_{id}(\tau) \geq (1+\tau)R_s(\tau/(1+\tau))$. By Theorem~\ref{th:: code rate subst} it holds that the rate $R_{id}(\tau)$, given in~\eqref{eq::main result}, is actually achievable and upper and lower bounds on the adversarial insertion-deletion channel coincide.

\section{Analysis of the feedback substitution code}\label{sec::analysis of the feedback substitution code}
In this section, we revisit the algorithm and the proof suggested by Zigangirov in~\cite{zigangirov1976number}. To make this section self-contained, we first describe an encoding process for the adversarial substitution channel with feedback. Then we introduce useful concepts and give a high-level analysis of this feedback strategy. Finally, we provide proofs of auxiliary technical statements used for showing the main result. 
\subsection{Feedback substitution code}\label{sec: feedback subsitution code}
Suppose that at most $\tau n$ substitution errors can occur during $n$ transmissions and the sender wants to transmit one out of $M$ messages. If $0\le\tau< (3-\sqrt{5})/4$, then the encoder takes positive real values $\alpha$ and $\beta$ such that $\alpha+\beta = 2$ and $\alpha$ is sufficiently close to $2\tau$. If $(3-\sqrt{5})/4\le \tau\le 1/3$, then $\alpha$ and $\beta$ are taken in such a way that $\alpha+\beta = 2$ and $\alpha$ is sufficiently close to and less than $(3-\sqrt{5})/2$. Clearly, for this assignment, we have that
\begin{equation}\label{eq::key inequality}
    \alpha\beta^2\leq 1.
\end{equation}
\textbf{Encoding procedure:} Initially, the sender takes the $[0,1]$-segment and partitions it into $M$ subsegments of equal length. Subsegments are enumerated from left to right. The $j$th message, $j\in[M]$, is associated  with the $j$th segment. Suppose that the encoder wants to transmit a message $m\in[M]$. The $m$th segment is then called the \textit{true}
segment and we denote it by $T(i)$ after the $i$th transmission. We define the union of the segments on the left and on the right to $T(i)$ by $L(i)$ and $R(i)$, respectively. The lengths of $L(i)$, $T(i)$ and $R(i)$ are denoted by $l(i)$, $t(i)$ and $r(i)$. Clearly, $t(0)=1/M$ and $l(i)+t(i)+r(i)=1$. At the $i$th moment, the sender checks whether the center of $T(i-1)$ is less than $\frac{1}{2}$ and transmits a $``0"$ to the channel in that case. Otherwise the encoder transmits a $``1"$ to the channel. If a $``0"$ is received, then everything on the left to $\frac{1}{2}$ is enlarged by $\beta$ and everything on the right to $\frac{1}{2}$ is shrunk by $\alpha$. More formally, 
\begin{align*}
l(i) &:= \beta \min\left(\frac{1}{2}, l(i-1)\right) + \alpha \max\left(l(i-1)-\frac{1}{2},0\right),\\
r(i) &:= \alpha \min\left(\frac{1}{2}, r(i-1)\right) + \beta \max\left(r(i-1)-\frac{1}{2},0\right),\\
t(i) &:= 1 -l(i) - r(i)  .
\end{align*}
If a $``1"$ is received, then everything on the left to $\frac{1}{2}$ is shrunk by $\alpha$ and everything on the right to $\frac{1}{2}$ is enlarged by $\beta$. Thus, 
\begin{align*}
l(i) &:= \beta \max\left(l(i-1)-\frac{1}{2},0\right) + \alpha \min\left(\frac{1}{2}, l(i-1)\right),\\
r(i) &:= \alpha \max\left(r(i-1)-\frac{1}{2},0\right) + \beta \min\left(\frac{1}{2}, r(i-1)\right),\\
t(i) &:= 1 -l(i) - r(i).
\end{align*}
Informally, given $\epsilon>0$, for properly chosen $\alpha\approx 2\min(\tau,(3-\sqrt{5})/4)$, sufficiently large $n$ and $M=2^{(R_s{(\tau)}-\epsilon)n}$, we claim that $T(n)$ contains the point $\frac{1}{2}$.
\subsection{Preliminaries}
Define $x(i):=\min(l(i), r(i))$ and $y(i) := \max(l(i), r(i))$. We say $T(i)$ is \textit{central} if $y(i) \le \frac{1}{2}$ and we say that there is a \textit{crossing} between moment $i-1$ and moment $i$
if either $(1)$ $l(i-1) \le r(i-1)$ and $l(i) > r(i)$ or $(2)$ $l(i-1) > r(i-1)$ and $l(i) \le r(i)$. We say $T(i)$ is
\textit{balanced} if $y(i)/x(i) \le \beta/\alpha$. We refer to Figure~\ref{fig:levelsets-props} for an illustration of the properties of $T(i)$.

\begin{figure}[h!]
  \centering
  \resizebox{.5\textwidth}{!}{
    \input{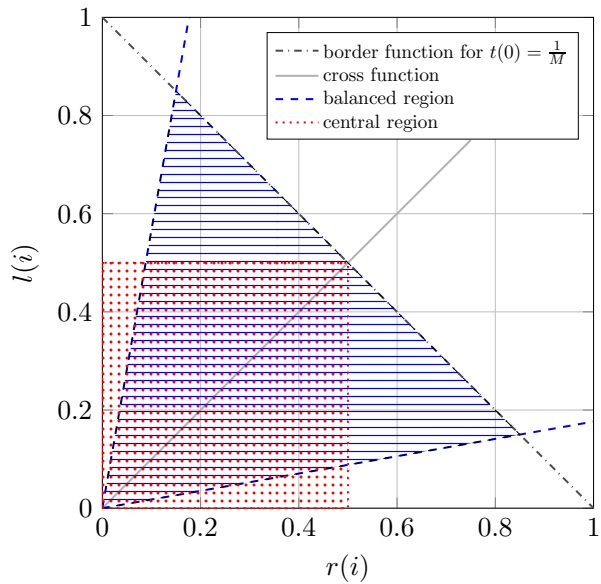}
    }
  \caption{Illustration for the notions of a crossing, a balanced and central true segment and the corresponding areas when $\tau = 0.15, \alpha = 2 \tau$, and $\beta = 2- \alpha$.}
    \label{fig:levelsets-props}
\end{figure}
\begin{example}
In Figure~\ref{fig:interval-change-transistion}, one can see how the segments $L(i), R(i)$, and $T(i)$ with $l(i)=0.2$, $r(i)=0.4$ and $t(i)=0.4$ change for the case of a correct and incorrect transmission when using the parameters $\tau=0.15$, $\alpha = 2\tau$, and $\beta = 2 -\alpha$. In particular, we see that the true segment $T(i)$ is the one containing the point $\frac{1}{2}$ and, thus, is central. Moreover, we see that $x(i)=l(i)$ and $y(i)=r(i)$, and can conclude that $T(i)$ is also balanced as  $2=\frac{y(i)}{x(i)} \leq \frac{\beta}{\alpha}=\frac{1.7}{0.3}$.
In the example, a correct transmission is sending a "0", i.e., enlarging the intervals on the left of $\frac{1}{2}$ by $\beta$ and shrinking the right hand side by $\alpha$. We can see that there was a crossing, since $l(i) \le r(i)$ and now $l(i+1) > r(i+1)$. Moreover, 
the true segment at moment $i+1$  still covers  the point $\frac{1}{2}$ and $T(i+1)$ is still balanced. In contrast, if the transmission was incorrect, i.e., enlarging by $\beta$ the right hand side of $\frac{1}{2}$ and shrinking the other part by $\alpha$, we observe that there was no crossing and the true segment $T(i+1)$ does not contain the point $\frac{1}{2}$ anymore. Moreover, we have that $T(i+1)$ is not balanced, since $\frac{34}{3} =\frac{y(i)}{x(i)} > \frac{\beta}{\alpha}=\frac{1.7}{0.3}$.
\end{example}

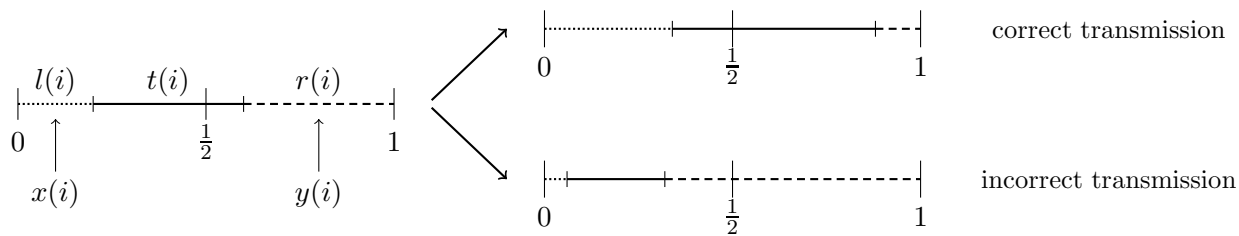
\begin{figure}[h!]
    	 \begin{tikzpicture}
		
		\def\xdist{7}
		\def\ydist{1}
		\def\leng{0.5}
		
		\coordinate (zero) at (0,0);

		\draw[thick,color=black,densely dotted] (0,0)--(2*\leng,0);
		\draw[thick,color=black,solid] (2*\leng,0)--(6*\leng,0);
		\draw[thick,color=black,densely dashed] (6*\leng,0)--(10*\leng,0);
		\foreach \x in {2*\leng,6*\leng}
		\draw (\x,3pt)--(\x,-3pt);
		
		\foreach \x in {0*\leng,5*\leng,10*\leng}
		\draw (\x,6pt)--(\x,-6pt);
		
		\node at (0*\leng,-0.5) {$0$};
		\node at (5*\leng,-0.5) {$\frac{1}{2}$};
		\node at (10*\leng,-0.5) {$1$};
		
		\node at (1*\leng,0.3) {$l(i)$};
		\node at (4*\leng,0.3) {$t(i)$};
		\node at (8*\leng,0.3) {$r(i)$};
		
		\node at (1*\leng,-1.2) {$x(i)$};
		\node at (8*\leng,-1.2) {$y(i)$};
		
		\draw[thin,->]  (1*\leng,-0.9) -- (1*\leng,-0.2);
		\draw[thin,->]  (8*\leng,-0.9) -- (8*\leng,-0.2);
		
		
		\draw[thick,color=black,densely dotted] (\xdist,\ydist)--($(\xdist,\ydist)+(3.4*\leng,0)$);
		\draw[thick,color=black,solid] ($(\xdist,\ydist)+(3.4*\leng,0)$)--($(\xdist,\ydist)+(8.8*\leng,0)$);
		\draw[thick,color=black,densely dashed] ($(\xdist,\ydist)+(8.8*\leng,0)$)--($(\xdist,\ydist)+(10*\leng,0)$);
		\foreach \x in {3.4*\leng,8.8*\leng}
		\draw ($(\xdist,\ydist)+(\x,3pt)$)--($(\xdist,\ydist)+(\x,-3pt)$);
		\foreach \x in {0*\leng,5*\leng,10*\leng}
		\draw ($(\xdist,\ydist)+(\x,7pt)$)--($(\xdist,\ydist)+(\x,-7pt)$);
		
		\node at ($(\xdist,\ydist)+(0,-0.5)$) {$0$};
		\node at ($(\xdist,\ydist)+(5*\leng,-0.5)$) {$\frac{1}{2}$};
		\node at ($(\xdist,\ydist)+(10*\leng,-0.5)$) {$1$};
		

		
		\draw[thick,color=black,densely dotted] (\xdist,-\ydist)--($(\xdist,-\ydist)+(0.6*\leng,0)$);
		\draw[thick,color=black,solid] ($(\xdist,-\ydist)+(0.6*\leng,0)$)--($(\xdist,-\ydist)+(3.2*\leng,0)$);
		\draw[thick,color=black,densely dashed] ($(\xdist,-\ydist)+(3.2*\leng,0)$)--($(\xdist,-\ydist)+(10*\leng,0)$);
		\foreach \x in {0.6*\leng,3.2*\leng}
		\draw($(\xdist,-\ydist)+(\x,3pt)$)--($(\xdist,-\ydist)+(\x,-3pt)$);
		\foreach \x in {0*\leng,5*\leng,10*\leng}
		\draw ($(\xdist,-\ydist)+(\x,7pt)$)--($(\xdist,-\ydist)+(\x,-7pt)$);
		
		\node at ($(\xdist,-\ydist)+(0,-0.5)$) {$0$};
		\node at ($(\xdist,-\ydist)+(5*\leng,-0.5)$) {$\frac{1}{2}$};
		\node at ($(\xdist,-\ydist)+(10*\leng,-0.5)$) {$1$};
		
		
		
		\draw[thick,->] (11*\leng,0.05*\ydist) -- ($(\xdist,\ydist)+(-1*\leng,0)$ );
		\draw[thick,->] (11*\leng,-0.05*\ydist) -- ($(\xdist,-\ydist)+(-1*\leng,0)$ );
		
		\node at ( $(\xdist,\ydist) + (15*\leng,0)$) {\small correct transmission};
		\node at ( $(\xdist,-\ydist) + (15*\leng,0)$) {\small incorrect transmission};

		\end{tikzpicture}
    \caption{Example for transition of the intervals $L(i), T(i)$, and $R(i)$ in case for a correct and incorrect transmission}
    \label{fig:interval-change-transistion}
\end{figure}

Define  the function
$$
g(i_0, i_1):= e(i_0, i_1)\log\alpha + f(i_0, i_1)\log\beta,
$$
where $e(i_0, i_1)$ and $f(i_0, i_1)$ are the numbers of incorrect and correct transmissions in moments $(i_0, i_1]=\{i_0+1,\ldots,i_1\}$. Note that the function $g$ is additive in the sense $g(i_0,i_2)=g(i_0,i_1)+g(i_1,i_2)$. For brevity, we write $g(n)$, $e(n)$ and $f(n)$ to denote $g(0,n)$, $e(0,n)$ and $f(0,n)$, respectively. We can think about the value $g(i,n)$ as of the amount of \textit{energy} we have at the moment $i$. For each correct transmission, we spend $\log\beta>0$, for an incorrect one, we get $-\log\alpha>0$.
\subsection{High-level analysis}
The following statement claims that if the amount of energy is sufficiently large, then at some moment the length of the true segment is bounded away from zero irrespective of the actual blocklength $n$ and we still have energy left to further increase the length of the true segment in the remaining steps. 
\begin{lemma}\label{lem::moment with big true segment}
For every $\varepsilon_1 > \varepsilon_2 > 0$ there exists $\delta>0$ such that for all but finitely many $\alpha \in (0,1)$ it holds: For $n \in \mathbb{N}$, the condition $g(0,n)>\varepsilon_1 n - \log t(0)$ implies that there is a first moment $n_1$ such that $t(n_1) \geq \delta$. Furthermore, $g(0, n_1) < \varepsilon_2 n - \log t(0)$ and $g(n_1,n)>(\epsilon_1-\epsilon_2)n$.
\end{lemma}
The following lemma shows that after spending enough energy the true segment will become central.
\begin{lemma} \label{lem::key lemma} Let $\alpha$ and $\beta$ satisfy~\eqref{eq::key inequality}. Then there exists a constant $c>0$ depending only on $\alpha,\beta$ and $t(n_1)$ such that if $g(n_1, n) > c$, it follows that $T(n)$ is central. Moreover, the length of the true segment at moment $n$ is $1-o(1)$ as $c\to\infty$.
\end{lemma}

Lemmas \ref{lem::moment with big true segment} and \ref{lem::key lemma} are proved in Section~\ref{sec::proofs of lemmas}.
Finally, we combine these two lemmas to obtain the main result.

\begin{theorem}\label{th::main}
    Let the maximal fraction of substitution errors be $\tau \in [0,1/3)$. For any $\varepsilon >0$, there exists $\alpha\in[0,(3-\sqrt{5})/2]$ such that for $n$ sufficiently large, the encoding procedure over the adversarial substitution channel with feedback can correct up to $\tau n$ errors with transmission rate at least $R_s(\tau) - \varepsilon$, where
    \begin{equation*}
	R_{s}(\tau)=\begin{cases}
	1 - h(\tau),\quad &\text{ for }0\le\tau\le (3-\sqrt{5})/4, \\
	(1-3\tau)\log\left(\frac{1+\sqrt{5}}{2}\right),\quad &\text{ for }(3-\sqrt{5})/4< \tau\le 1/3.
	\end{cases}
	\end{equation*}
\end{theorem}
\begin{proof}
 We choose $M$ and accordingly $t(0)=1/M$ such that it satisfies  $2^{(R_s(\tau)-\varepsilon)n}<M < 2^{(R_s(\tau)-\varepsilon/2)n}$ and  $-\log t(0) < (R_{s}(\tau)-\varepsilon/2)n$. By definition, we have that $e(0,n) \leq \tau n$ and $f(0,n) \geq (1-\tau) n$. For given $\alpha$ and $\beta$, we can compute
\begin{align*}
    g(0,n)=e(0,n) \log \alpha + f(0,n) \log \beta \geq (\tau \log \alpha + (1-\tau) \log \beta)n.
\end{align*}
Subject to the constraints $\alpha + \beta = 2$ and $\alpha \beta^2\le 1$, the function $\tau \log \alpha + (1-\tau) \log \beta$ reaches the maximum $R_s(\tau)$ when $ \alpha = 2\min(\tau,(3-\sqrt{5})/4)$ and $\beta = 2-\alpha$. To make use of Lemma~\ref{lem::moment with big true segment}, we take $\alpha$ slightly smaller than $2\min(\tau,(3-\sqrt{5})/4)$ to have
\begin{align*}
    g(0,n)=e(0,n) \log \alpha + f(0,n) \log \beta \geq \left(R_{s}(\tau) - \frac{\varepsilon}{10}\right)n > \frac{2}{5} \varepsilon n - 
    \log t(0).
\end{align*}
By applying Lemma~\ref{lem::moment with big true segment} with $\epsilon_1=2\epsilon/5$ and $\epsilon_2=\epsilon/5$, there exists a constant $\delta >0$ such that for all but finite $\alpha \in (0,1)$, there is a first moment $n_1$ satisfying $t(n_1) \geq \delta$ and furthermore
\begin{align*}
   g(n_1,n)=g(0,n)-g(0,n_1)= e(n_1,n) \log \alpha + f(n_1,n) \log \beta > \frac{1}{5} \varepsilon n.
\end{align*}
Thus, for $n$ sufficiently large, we conclude by Lemma~\ref{lem::key lemma} that $T(n)$ is central.
\end{proof}

The key ingredient in the proof of Theorem~\ref{th::main result} is Lemma~\ref{lem::key lemma}. In order to prove it we introduce several auxiliary functions
\begin{align*}
u_1(i)&:= \log(t(i)/x(i)),\\
v_1(i)&:=\log (2t(i)),\\
u_2(i)&:=
\begin{cases}
-\log (4x(i)y(i)),\quad &\text{if }y(i)\le \frac{1}{2},\\
\log((1-y(i))/x(i)),\quad &\text{if }y(i)> \frac{1}{2}, 
\end{cases}
\\
v_2(i)&:=
\begin{cases}
-\log (2y(i)),\quad &\text{if }y(i)\le \frac{1}{2},\\
\log (2(1-y(i))),\quad &\text{if }y(i)> \frac{1}{2}.
\end{cases}
\end{align*}
We also define the constants $u_1':=\log(\beta - \alpha)$, $u_1'':=\log(\beta/\alpha - 1)$, $u_2':=\log(\beta/ \alpha)$. Clearly, $u_1'<u_1''$.
Finally, define
\begin{align*}
u(i)&:=
\begin{cases}
u_1(i),\quad &\text{if }u_1(i)< u_1',\\
u_2(i),\quad &\text{if }u_1(i)\geq u_1', 
\end{cases}
\\
v(i)&:=
\begin{cases}
v_1(i),\quad &\text{if }u_1(i)< u_1'',\\
v_2(i),\quad &\text{if }u_1(i)\geq u_1''.
\end{cases}
\end{align*}
To show the validity of Lemma~\ref{lem::key lemma}, we prove that the function $v$ takes a large value at the last moment $n$, which happens if and only if $l(n)$ and $r(n)$ are small (c.f. Example~\ref{exmp:levelsets-u-v}). 

\begin{proposition}\label{prop::function v}
If $T(i-1)$ is central and there is a crossing between moments $(i-1)$ and $i$, then $v(i)-v(i-1)\geq -\log\beta$. Otherwise, 
\begin{align*}
v(i)-v(i-1)&\geq
\begin{cases}
\log\beta,\quad &\text{if transmission is correct},\\
\log\alpha,\quad &\text{if transmission is incorrect}. 
\end{cases}
\end{align*}
\end{proposition}

The change $\Delta v(i):=v(i)-v(i-1)$ of the function $v$ is equal to $g(i-1, i)$, except for the event that $T(i-1)$ is central and there is a crossing. So, if this event is only happening rarely, then $v(n)-v(n_1)$ is close to $g(n_1, n)$. Thus, because $v(n_1)$ can be lower bounded by a function of $t(n_1)$, $v(n)$ is large if $g(n_1,n)$ is large and we are done. To handle the case when those events appear more frequently, we use the function $u$ and the following proposition.

\begin{proposition}\label{prop::function u}
If $T(i-1)$ is central and there is a crossing, then $u(i)-u(i-1)\geq \log\beta$. Otherwise, $u(i)-u(i-1)\geq 0$.
\end{proposition}

So, in this case the function $u$ is big, which implies that the function $u_2$ is also big. However, this is not enough to conclude that $T(n)$ is central (c.f. Example~\ref{exmp:levelsets-u-v}). Consider the last occurrence of a crossing with the true segment having been central before the crossing. At this moment, the true segment can be shown to be also balanced, which means that the function $v_2$ is close to $u_2$, i.e., it is also big. To make $v_2$ small again, we must have a lot of incorrect transmissions due to Proposition~\ref{prop::function v2}.

\begin{proposition}\label{prop::function v2}
    If there is a crossing and $u_1(i-1) \geq u_1''$, then $v_2(i) - v_2(i-1) \geq - \log \beta$. If there is no crossing, then
    \begin{equation*}
        v_2(i)-v_2(i-1) \geq \begin{cases}
            \log \beta, \quad \text{if the transmission is correct},\\
            \log \alpha, \quad \text{if the transmission is incorrect and } v_2(i-1) < \log \beta,\\
            -\log \beta \ (\geq \log \alpha), \quad \text{if the transmission is incorrect and } v_2(i-1) \geq \log \beta.
        \end{cases}
    \end{equation*}
\end{proposition}

Plenty of incorrect transmissions implies that we have a lot of energy left. Now we can use again Proposition~\ref{prop::function v}, since it is not possible that a crossing occurs when the true segment is central and thus, conclude that the function $v$ takes a large value after $n$ transmissions.

\begin{figure}
\centering
\begin{minipage}{.49\textwidth}
  \centering
    \input{rainbow_v_funcs}
  \caption{Illustration for level sets of the function $v(i)$ for $\tau = 0.15$, $\alpha = 2 \tau$, and $\beta = 2- \alpha$.}
    \label{fig:levelsets-v-func}
\end{minipage}
\hfill
\begin{minipage}{.49\textwidth}
  \centering
    \input{rainbow_u_funcs}
     \caption{Illustration for level sets of the function $u(i)$ for $\tau = 0.15$, $\alpha = 2 \tau$, and $\beta = 2- \alpha$.}
    \label{fig:levelsets-u-func}
\end{minipage}
\end{figure}

\begin{example}\label{exmp:levelsets-u-v}
Figures \ref{fig:levelsets-v-func} and \ref{fig:levelsets-u-func} illustrate different level sets of the functions $v(i)$ and  $u(i)$ if they are treated as function on $l(i)$ and $r(i)$. A level set of $v(i)$ is a set of pairs $(l(i),r(i))$ such that $v(i) = const$. Note that we use the same color for the same level sets in the graphs. For high level sets of the function $v(i)$, we are close to the origin of the graph, which corresponds to $l(i)$ and $r(i)$ being small and the $t(i)$ very large. In contrast, for high level sets of the function $u(i)$, we can only claim that we are close to an axis, i.e., either $l(i)$ or $r(i)$ is small. Notice that in case $T(i)$ is central, we need to have $l(i),r(i) < \frac{1}{2}$. Combining the two figures for the level sets of $v(i)$ and $u(i)$, it becomes clear that we need to prove that the function $v(n)$ is large in order to conclude that after $n$ transmissions $T(n)$ is central.
\end{example}

Many technical details were hidden in this discussion. In order to give a formal proof of Lemma~\ref{lem::key lemma}, we need some additional propositions.

\begin{proposition}\label{prop::trivial stuff} We have the following (trivial) properties.
    \begin{enumerate}
        \item If there is a crossing, then the transmission is correct.\label{it::cross implies correct transmission}
        \item If the transmission is correct, then the length of the true segment is increased, i.e., $t(i) > t(i-1)$\label{it::correct transmission increases true segment}, otherwise it is decreased.
        \item If $T(i-1)$ is central, $T(i)$ is not central and there is no crossing, then it follows that the transmission was incorrect.\label{it::central to non-central without cross implies error}
        \item The following inequalities always hold $\alpha t(i-1)\leq t(i)\leq \beta t(i-1)$.
        \label{it::limits for t(i)}
        \item For any $x(i)$, $y(i)$ and $t(i)$, inequalities $u_2(i)\geq u_1(i)$ and $u_2(i)\geq u(i)$ hold.\label{it::u2>u1}
        \item For any $x(i)$, $y(i)$ and $t(i)$, inequalities $v_2(i)\geq v_1(i)$ and $v_2(i)\geq v(i)$ hold.\label{it::v2>v1}
    \end{enumerate}
\end{proposition}

\begin{proposition}\label{prop::properties of functions}
	We have the following properties.
	\begin{enumerate}
		\item If $T(i)$ is central, then $u_2(i)\ge 2 v_2(i)$. \label{it::is central then}
		\item If $T(i)$ is central and balanced, then $u_2(i) < 2v_2(i)+\log(\beta/\alpha)$. \label{it:: is central and balanced}
		\item If there is a crossing between moment $i-1$ and moment $i$, then $T(i-1)$ and $T(i)$ are balanced. \label{it:: cross, then balanced}
		\item If $u_1(i-1) < u_1''$  (or $u_1(i-1) < u_1'$)  and $T(i-1)$ is central, and the transmission is correct, then there is a crossing between moment $i-1$ and moment $i$. \label{it:: u1 is small and was central and correct transmission, then cross}
		\item If $u_1(i-1)\ge u_1'$ (or $u_1(i-1)\ge u_1''$) and there is a crossing between moment $i-1$ and moment $i$, then $T(i)$ is central. \label{it:: u1 is large and cross, then is central}
		\item If $u_1(i-1)\ge u_1''$ and there is a crossing between moment $i-1$ and $i$, then $T(i-1)$ is central and $T(i)$ is central. \label{it:: u1 is large and cross, then was central}
		\item If $u_2(i)\ge u_2'$ and there is either a crossing between moment $i-1$ and moment $i$, or a crossing between moment $i$ and moment $i+1$, then $T(i)$ is central. \label{it:: u2 is large and cross, then is central}
	\end{enumerate}
\end{proposition}

\begin{proposition}\label{prop::u1increasing}
    If $u_1(i-1) < u_1'$, then $u_1(i) - u_1(i-1) \geq 0$. If in addition $T(i-1)$ is central and there is a crossing, then $u_1(i) - u_1(i-1) \geq \log \beta$. On the other hand, if $u_1(i-1) \geq u_1'$, then $u_1(i) \geq u_1'$, and if $u_1(i-1) \geq u_1''$, then $u_1(i) \geq u_1''$.
\end{proposition}

\begin{proposition}\label{prop::u2increasing}
    If $u_1(i-1) \geq u_1'$, then $u_2(i) - u_2(i-1) \geq 0$. If both $T(i-1)$ and $T(i)$ are central, then $u_2(i) - u_2(i-1) = -\log(\alpha \beta)$.
\end{proposition}

\begin{proposition}\label{prop::function v1}
    Suppose $u_1(i-1) < u_1''$. If there is a crossing, then $v_1(i) - v_1(i-1) > 0$. Otherwise,
    \begin{equation*}
        v_1(i) - v_1(i-1) \geq \begin{cases}
            \log \beta, \quad \text{if the transmission is correct},\\
            \log \alpha, \quad \text{if the transmission is incorrect}.
        \end{cases}
    \end{equation*}
\end{proposition}

\begin{proposition}\label{prop::u2 and g}

There exists a constant $c=-5\log(\alpha\beta)$ such that if both $T(i_0)$ and $T(i_1)$ are central and balanced, and $u_2(i_0)\geq c$ then
\begin{equation}\label{eq::g and u_2}
    g(i_0, i_1)\leq u_2(i_1)-u_2(i_0).
\end{equation}
\end{proposition}

The proof of lemmas and proposition are given in Section~\ref{sec::proofs of lemmas} and~\ref{sec::proofs of propositions}, respectively. At last, we provide a flowchart in Figure~\ref{fig:flowchart} to make it easier to understand the scheme of the proof.

\begin{figure}[h!]
    \centering
    \resizebox{0.7\textwidth}{!}{
    	\begin{tikzpicture}[auto, inner sep=0,
		block_eq/.style ={rectangle, draw=black, thick, fill=light-gray,
			text width=5em, text centered, minimum height=3em, inner sep=6pt,text opacity = 1},
		block_prop4/.style ={rectangle, draw=black, thick, fill=dark-gray,
			text width=5em, text centered, minimum height=3em, inner sep=6pt},
		block_prop5/.style ={rectangle, draw=black, thick, fill=white,
			text width=5em, text centered, minimum height=3em, inner sep=6pt,text opacity = 1},
		block_no/.style ={rectangle, draw=black, thick, fill=white,
			text width=5em, text centered, minimum height=3em, inner sep=6pt},
		block_lem/.style ={rectangle, draw=black, thick, fill=white,
			text width=5em, text centered, minimum height=3em, inner sep=6pt,text opacity = 1},
		block_thm/.style ={rectangle, draw=black, thick, fill=white,
			text width=5em, text centered, minimum height=3em, inner sep=6pt},
		block_red/.style ={rectangle, draw=black, thick, fill=light-gray,
			text width=1.5em, text centered, minimum height=0.5em, inner sep=0pt},
		block_blue/.style ={rectangle, draw=black, thick, fill=dark-gray,
			text width=1.5em, text centered, minimum height=0.5em, inner sep=0pt},
		block_green/.style ={rectangle, draw=black, thick, fill=white,
			text width=1.5em, text centered, minimum height=0.5em, inner sep=0pt},
		]
		
		\def\xdist{3};
		\def\ynode{1};
		
		\node[block_no] (prop6) at (0,0) {Prop.~\ref{prop::u1increasing} };
		\node[block_red,xshift=-2em,yshift=-1em] at (prop6) {};
		\node[block_blue,xshift=0em,yshift=-1em] at (prop6) {};
		\node[block_green,xshift=2em,yshift=-1em] at (prop6) {};
		
		\node[block_no] (prop7) at (0,-3*\ynode) {Prop.~\ref{prop::u2increasing} };
		\node[block_blue,xshift=0em,yshift=-1em] at (prop7) {};

		
		\node[block_no] (prop2) at (1.2*\xdist,-\ynode) {Prop.~\ref{prop::function u} };
		\node[block_red,xshift=-2em,yshift=-1em] at (prop2) {};
		\node[block_blue,xshift=0em,yshift=-1em] at (prop2) {};
		\node[block_green,xshift=2em,yshift=-1em] at (prop2) {};
		
		\node[block_no] (prop9) at (1.2*\xdist,-3*\ynode) {Prop.~\ref{prop::u2 and g} };
		\node[block_red,xshift=-2em,yshift=-1em] at (prop9) {};
		\node[block_blue,xshift=0em,yshift=-1em] at (prop9) {};
		\node[block_green,xshift=2em,yshift=-1em] at (prop9) {};

		\node[block_no] (prop8) at (1.2*\xdist,-5*\ynode) {Prop.~\ref{prop::function v1} };
		\node[block_blue,xshift=0em,yshift=-1em] at (prop8) {};
		\node[block_green,xshift=2em,yshift=-1em] at (prop8) {};
		
		\node[block_no] (prop3) at (1.2*\xdist,-7*\ynode) {Prop.~\ref{prop::function v2} };
		\node[block_blue,xshift=0em,yshift=-1em] at (prop3) {};

		\node[block_no] (prop1) at (2.5*\xdist,-5.2*\ynode) {Prop.~\ref{prop::function v} };
		\node[block_blue,xshift=0em,yshift=-1em] at (prop1) {};
		\node[block_green,xshift=2em,yshift=-1em] at (prop1) {};

		\node[block_eq] (ineq3) at (-0*\xdist,-5.5*\ynode) {Eq.~\eqref{eq::key inequality} };
		\node[block_prop4] (prop4) at (-0*\xdist,-6.75*\ynode) {Prop.~\ref{prop::trivial stuff} };
		\node[block_prop5] (prop5) at (-0*\xdist,-8*\ynode) {Prop.~\ref{prop::properties of functions} };
		
		\node[block_lem] (lem5) at (4*\xdist,-3*\ynode) {Lemma~\ref{lem::moment with big true segment}};
		\node[block_lem] (lem6) at (4*\xdist,-5*\ynode) {Lemma~\ref{lem::key lemma}};
		\node[block_red,xshift=-2em,yshift=-1em] at (lem6) {};
		\node[block_blue,xshift=0em,yshift=-1em] at (lem6) {};
		\node[block_green,xshift=2em,yshift=-1em] at (lem6) {};
		
		\node[block_thm] (thm7) at (5*\xdist,-4*\ynode) {Theorem~\ref{th::main}};

		\draw[->,thick] ($(prop6.south east)!0.75!(prop6.north east)$) -- +(4.2,0) |- ($(prop1.south west)!0.75!(prop1.north west)$);
		
		\draw[->,thick] ($(prop7.south east)!0.75!(prop7.north east)$) -- +(0.2,0) |- ($(prop2.south west)!0.25!(prop2.north west)$);
		
		\draw[->,thick] ($(prop7.south east)!0.25!(prop7.north east)$) -- ($(prop9.south west)!0.25!(prop9.north west)$);
		
		\draw[->,thick] ($(prop6.south east)!0.5!(prop6.north east)$) -- +(0.7,0) |- ($(prop2.south west)!0.75!(prop2.north west)$);
		
		\draw[->,thick] ($(prop6.south east)!0.25!(prop6.north east)$) -- +(0.5,0) |- ($(prop9.south west)!0.75!(prop9.north west)$);

		\draw[->,thick] (prop8.east) -- +(0.3,0) |- ($(prop1.south west)!0.5!(prop1.north west)$);
		
		\draw[->,thick] ($(prop3.south east)!0.75!(prop3.north east)$) -- +(0.8,0) |- ($(prop1.south west)!0.25!(prop1.north west)$);
		
		\draw[->,thick] (prop2.east) -- +(5,0) |- ($(lem6.south west)!0.8!(lem6.north west)$);
		
		\draw[->,thick] (prop9.east) -- +(4.25,0) |- ($(lem6.south west)!0.6!(lem6.north west)$);
		
		\draw[->,thick] (prop1.east) -- +(0.6,0) |- ($(lem6.south west)!0.4!(lem6.north west)$);


		\draw[->,thick] ($(prop3.south east)!0.25!(prop3.north east)$) -- +(4.5,0) |- ($(lem6.south west)!0.20!(lem6.north west)$);

		\draw[->,thick] (lem5.east) -- +(0.4,0) |- ($(thm7.south west)!0.75!(thm7.north west)$);
		
		\draw[->,thick] (lem6.east) -- +(0.4,0) |- ($(thm7.south west)!0.25!(thm7.north west)$);


		\end{tikzpicture}
    }
    \caption{Illustration of the proof strategy for Theorem~\ref{th::main}. For the inequality~\eqref{eq::key inequality}, Proposition~\ref{prop::trivial stuff}, or Proposition~\ref{prop::properties of functions} the arrows are omitted due to the sake of clarity. The small indicator boxes however show where they are used.}
    \label{fig:flowchart}
\end{figure}
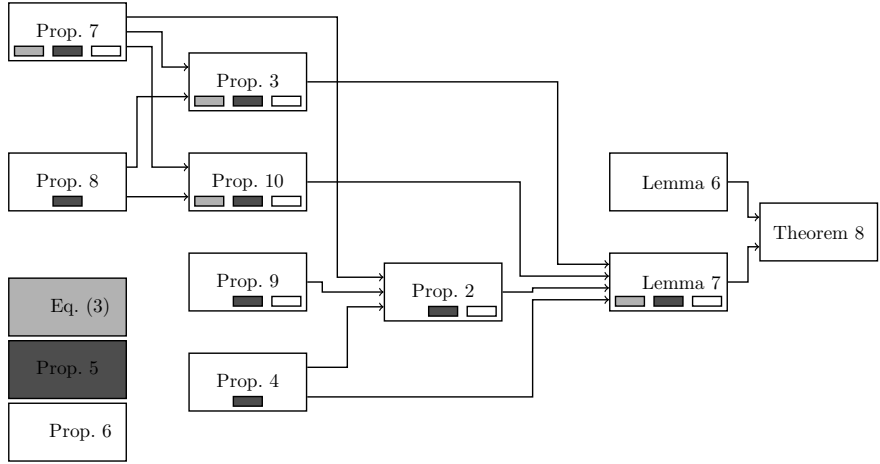

In prior of the technical proofs of the propositions, we provide in Example~\ref{examp:trajectory} an illustration for an example transmission, where we track the length of the true segment $T(i)$ and its position.

\begin{figure}[h!]
    \centering
    \resizebox{.7\textwidth}{!}{
    \input{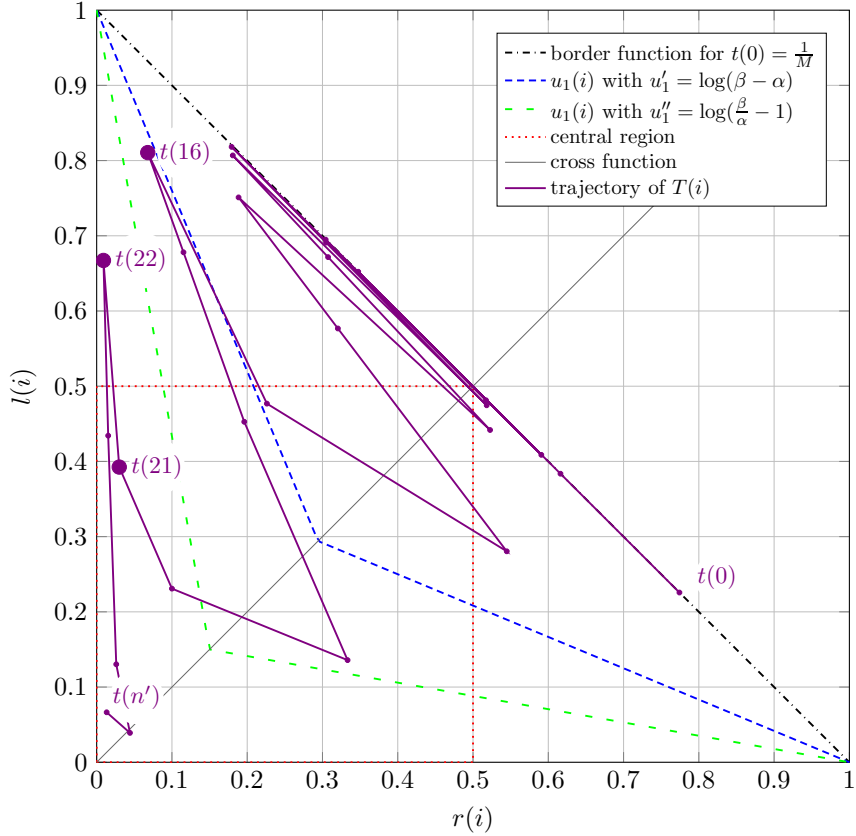}
    }
    \caption{Illustration for a trajectory of $t(i)=1-l(i)-r(i)$. We have $M = 2^{14}$ and transmit a random message, where we use $n^\prime = 27$ total transmissions. Errors occur at time steps 16, 21, and 22. Used parameters are $\tau = 0.15$, $\alpha = 2 \tau$, and $\beta = 2- \alpha$.}
    \label{fig:trajectory}
\end{figure}

\begin{example}\label{examp:trajectory}
We illustrate in Figure~\ref{fig:trajectory} an example transmission of a random message $m$ out of $M=2^{14}$ for $n'=27$ total transmission. We use that $\tau = 0.15$, $\alpha = 2\tau$ and $\beta = 2- \alpha$. Recall that $t(i) = 1-l(i)-r(i)$. We have three bit-flip errors at the time steps $i\in\{16,21,22\}$. Since at the first transmission all segments are the same, thus the position of the true segment $T(0)$ described by the pair $(l(0),r(0))$ lies close to the line $l(i)+r(i)=1$. For more and more bit transmissions, we see that the trajectory of $T(i)$ is getting closer to the origin, meaning that its length is increasing. We have that $T(i)$ is central at time step $i=15$ after multiple crosses before. Since at time step $i=16$ an error occurred, the trajectory leaves the central region. Furthermore, at the same time step, the trajectory crosscuts the level set $u_1^\prime$ and does not cross it again in later time steps (c.f. Proposition~\ref{prop::u1increasing}). The same behavior we observe if the level set $u_1^{\prime \prime}$ is crosscut although we have errors at time steps $i=21,22$, which makes $T(i)$ not central anymore (c.f. Proposition~\ref{prop::u1increasing}). For the remaining correct transmissions, we see that $t(i)$ is always increasing until we reach the last $n'$ transmission.
\end{example}

\subsection{Proofs of lemmas}\label{sec::proofs of lemmas}

\begin{proof}[Proof of Lemma~\ref{lem::moment with big true segment}]
We choose $k \in \mathbb{N}$ such that
$k > \frac{\log \beta}{\varepsilon_2}$ for $\varepsilon_1>\varepsilon_2 > 0$. Let us fix $\delta < \frac{1}{\beta ^2}$ which corresponds to the smallest displacement of the position $\frac{1}{2}$ after $1,2,\ldots,k-1$ transmissions. For instance, after transmitting one symbol the point $\frac{1}{2}$ can be mapped to one of the points $\{\frac{\alpha}{2},\frac{\beta}{2}\}$, whereas after two transmissions the image of the point $\frac{1}{2}$ is one of the points $\{\frac{\alpha^2}{2}, \frac{\alpha\beta}{2},1-\alpha(1-\frac{\beta}{2}), 1-\beta(1-\frac{\beta}{2})\}$. One can prove that $\delta > 0$ for all but finite $\alpha \in (0,1)$. Remark that if $\alpha$ is a non-algebraic number, then $\delta>0$ for any $k$.

Assume on the contrary that $t(i) < \delta$ for $i \in \{0,1,\ldots,n\}$. We have by the definition of $\delta$ that if $T(i)$ is central, i.e., the point $\frac{1}{2}$ is in $T(i)$, then $T(i+1),\ldots,T(i+k-1)$ are not central. Let us denote as $n^\prime$ the number of moments when $T(i)$ is central. By the definition we know that $n^\prime \leq \frac{n}{k}+1 < \frac{\varepsilon_1 n }{\log \beta} + 1$. Therefore, we conclude
$$
t(n) \geq t(0) \alpha ^{e(0,n)} \beta ^{f(0,n)-n^\prime} \geq t(0) \alpha ^{e(0,n)} \beta ^{f(0,n) - \frac{\varepsilon_1 n}{\log \beta} - 1} \geq \frac{1}{\beta} > \delta,
$$
where we used the fact that $g(0,n)=e(0,n)\log \alpha + f(0,n)\log\beta > \epsilon_1 n - \log t(0)$. Hence, we come to a contradiction. 

Let $n_1$ be a first moment when $t(n_1)\ge \delta$. Applying the same strategy we can see that
$$
\beta \delta > \beta t(n_1 - 1) \geq t(n_1) \geq t(0) \alpha ^{e(0,n_1)} \beta ^{f(0,n_1)-\frac{\varepsilon_2 n_1 }{\log \beta} - 1}.
$$
From the two aforementioned observations it follows that
$$
g(0,n_1)=e(0,n_1)\log \alpha + f(0,n_1) \log \beta < \varepsilon_2 n_1 - \log t(0) + \log( \beta ^2 \delta) < \varepsilon_2 n - \log t(0) .
$$
Hence
$$
g(n_1,n)=g(0,n)-g(0,n_1)>(\epsilon_1-\epsilon_2)n>0.
$$
\end{proof}

\begin{proof}[Proof of Lemma~\ref{lem::key lemma}]
    Set $I:= \{i\in (n_1,n] : T (i-1)$ is central and there is a crossing between moments $i - 1$ and $i\}$. Let $m :=
    |I|$, and take $c  \gg m' \gg v_2' \gg 1$. One way to think about the choices of the parameters is to let
    $c = x^3, m' = x^2 , v_2'= x$, where $x$ is sufficiently large relative to $\alpha,\beta,t(n_1)$.\\
    \textbf{Case} $m\le m'$. From Proposition~\ref{prop::trivial stuff}.\ref{it::v2>v1} and the definition of $v$ it follows that $v(n_1)\ge v_1(n_1)= \log (2 t(n_1))$. Thus, by Propositions~\ref{prop::function v}, we have
    \begin{align*}
    v(n)&\ge v(n_1)+e(n_1,n)\log\alpha + (f(n_1,n)-m) \log \beta - m \log \beta 
    \\
    &\ge \log (2t(n_1)) + g(n_1,n) - 2m' \log \beta 
    \\
    &> \log (2t(n_1)) + c - 2m' \log \beta
    \gg 1.
    \end{align*}
    Here we used our choice of $c$ and $m'$ and the fact that $g(n_1,n)=e(n_1,n)\log\alpha +f(n_1,n)\log \beta$. From the definition of function $v$ and the fact that $v(n)\gg 1$
    we conclude that $T(n)$ is central (see Fig.~\ref{fig:levelsets-v-func}).\\
    \textbf{Case} $m>m'$. Let $n_2$ be the $m'$th smallest element of $I$. From Proposition~\ref{prop::trivial stuff}.\ref{it::u2>u1} and the definition of function $u$ it follows that  $u_2(n_2)\ge u(n_2)$, and $u(n_1)\ge u_1(n_1)=\log (t(n_1)/x(n_1)))\ge \log (2 t(n_1))$ as $x(n_1)\le 1/2$. Combining this together with Proposition~\ref{prop::function u} we get
    $$
    u_2(n_2)\ge u(n_2)\ge u(n_1) + m' \log \beta \ge \log (2t(n_1)) + m'\log \beta = \Omega(m').
    $$
    The last equality holds due to our choice of $m'$.
Thus, by Proposition~\ref{prop::properties of functions}.\ref{it:: u2 is large and cross, then is central}, $T(n_2)$ is central and by Proposition~\ref{prop::properties of functions}.\ref{it:: cross, then balanced}, $T(n_2)$ is balanced.
    By Propositions~\ref{prop::function v},~\ref{prop::trivial stuff}.\ref{it::v2>v1} and~\ref{prop::properties of functions}.\ref{it::is central then}, we know that
    \begin{align}
    g(n_1,n_2)&\le v(n_2)-v(n_1)+2m'\log\beta \le v_2(n_2)-v_1(n_1)+2m'\log\beta \nonumber\\
    &\le \frac{1}{2} u_2(n_2) - \log (2 t(n_1)) + 2m'\log \beta \le \frac{1}{2} u_2(n_2) + O(m').\label{eq::g(0,n1)}
    \end{align}
    Let $n_3$ be the last element of $I$. We know that $T(n_2)$ is central and balanced, and $u_2(n_2)= \Omega(m')$. The same conclusion holds for $T(n_3)$, in particular, $T(n_3)$ is central and balanced, and $u_2(n_3)= \Omega(m')$. By Proposition~\ref{prop::u2 and g}, we obtain
    \begin{equation}\label{eq::g(n1,n2)}
    g(n_2,n_3) \le u_2(n_3) - u_2(n_2).
    \end{equation}
    By Proposition~\ref{prop::properties of functions}.\ref{it:: is central and balanced}, we derive that
    $$
    v_2(n_3)\ge \frac{1}{2}u_2(n_3) - \frac{1}{2}\log(\beta/\alpha) = \Omega(m')\gg v_2'.
    $$
    If $v_2(n)\ge v_2'$, we are done. Otherwise, let $n_4> n_3$ be a first moment when $v_2(n_4)<v_2'$. Since $v_2(i)>0$ for all $i\in[n_3,n_4)$, from the definition of the function $v_2$ it follows that $T(i)$ is central for all $i\in[n_3,n_4)$. By the definition of $I$, there are no crossings between $n_3$ and $n_4$. By Propositions~\ref{prop::function v2} and~\ref{prop::properties of functions}.\ref{it:: is central and balanced} and the fact $\alpha \beta < 1$, we obtain
    \begin{align}
    g(n_3,n_4)&=e(n_3,n_4)\log\alpha + f(n_3,n_4)\log\beta\le 2(-e(n_3,n_4)+f(n_3,n_4))\log \beta \nonumber \\
    &\le 2 v_2(n_4)- 2 v_2(n_3) < -u_2(n_3) + O(v_2'). \label{eq::g(n2,n3)}
    \end{align}
    By the definition of $I$, there is no $i\in(n_4,n]$ such that $T(i-1)$ is central and there is a crossing between moments $i-1$ and $i$. By Proposition~\ref{prop::function v}, we have
    \begin{equation}\label{eq::g(n3,n)}
    g(n_4,n)\le v(n) - v(n_4).
    \end{equation}
    Since there is no crossing between $n_4-1$ and $n_4$, by Proposition~\ref{prop::function v2}, we have $v_2(n_4)\ge v_2(n_4-1)-\log \beta \ge v_2' - \log \beta \gg 1$. By the definition of $v_2(n_4)$, we know that $y(n_4)\ll 1$ and so $v_1(n_4)> 0$, hence $v(n_4)\ge 0$. Adding up~\eqref{eq::g(0,n1)}-\eqref{eq::g(n3,n)} and the fact that $v(n_4)\ge 0$ yield
    $$
    v(n)\ge g(n_1,n)-O(m')-O(v_2')\gg 1,
    $$
    which implies that $T(n)$ is central.
\end{proof}

\subsection{Proofs of propositions}\label{sec::proofs of propositions}

\begin{proof}[Proof of Proposition~\ref{prop::trivial stuff}]
   
    \ref{it::cross implies correct transmission}: If there is a crossing, then the smaller segment among $\{L(i), R(i)\}$ has been increased. It happens only when the transmission is correct.\\
    \ref{it::correct transmission increases true segment}: We only proof the assertion for correct transmission because the case for incorrect transmission is proved analogously.
    In case that $T(i-1)$ is non-central the statement follows because $t(i) = t(i-1) \beta$.
    In case $T(i-1)$ is central we have that
    \begin{equation*}
        t(i) = 1-\beta x(i-1) - \alpha y(i-1)
    \end{equation*}
    and therefore
    \begin{equation*}
        t(i) - t(i-1) = (1-\beta) x(i-1) + (1-\alpha)y(i-1) = \frac{\beta - \alpha}{2} (y(i-1) - x(i-1))>0,
    \end{equation*}
    where we have used that $\alpha + \beta = 2$.\\
    \ref{it::central to non-central without cross implies error}: We assume that the transmission is correct and there is no crossing. Then since $T(i-1)$ is central it holds that
    \begin{align}\label{eq:central_to_non_central}
        x(i-1) \beta &= x(i),\nonumber\\
        y(i-1) \alpha &= y(i).
    \end{align}
    In order for $T(i)$ to be non central it is necessary that $y(i)>\frac{1}{2}$ but this leads to a contradiction due to \eqref{eq:central_to_non_central} and because $\alpha<1$ and $y(i-1)<\frac{1}{2}$.\\
    \ref{it::limits for t(i)}: If $T(i-1)$ is not central its length is multiplied by $\beta$ for a correct transmission and by $\alpha$ for an erroneous one. Because the length of all subsegments combined remains one during the entire process, the length of the true segment is multiplied by a value in $[\alpha,\beta]$ if $T(i-1)$ is central, completing the proof.\\
    \ref{it::u2>u1}: Recall that
    \begin{equation*}
        u_2(i):= 
        \begin{cases}
            -\log (4x(i)y(i)),\quad &\text{if }y(i)\le \frac{1}{2},\\
            \log((1-y(i))/x(i)),\quad &\text{if }y(i)> \frac{1}{2}.
        \end{cases}
    \end{equation*}
    We claim that
    \begin{equation*}
        -\log(4x(i)y(i)) \geq \log((1-y(i))/x(i))  .
    \end{equation*}
    The following statements are equivalent.
    \begin{align*}
        \frac{1}{4x(i)y(i)} &\geq \frac{1-y(i)}{x(i)}, \nonumber\\
        1 &\geq 4y(i)(1-y(i)).
    \end{align*}
    The last inequality follows because its right hand side is maximized for $y(i)=\frac{1}{2}$. This maximal value is equal to one and our claim follows.
    It follows that $u_2(i) \geq u_1(i)=\log(t(i)/x(i))$ because
    \begin{equation*}
        \frac{1-y(i)}{x(i)} \geq \frac{t(i)}{x(i)}  = 2^{u_1(i)}  .
    \end{equation*}
    \ref{it::v2>v1}: Recall that
    \begin{equation*}
        v_2(i):=
        \begin{cases}
            -\log (2y(i)),\quad &\text{if }y(i)\le \frac{1}{2} ,\\
            \log (2(1-y(i))),\quad &\text{if }y(i)> \frac{1}{2}  .
        \end{cases}
    \end{equation*}
    Similar to the proof of Proposition~\ref{prop::trivial stuff}.\ref{it::u2>u1} we show first that
    \begin{equation}\label{ineq:v2>v1}
        -\log(2y(i)) \geq \log(2(1-y(i)))  .
    \end{equation}
    The following statements are equivalent.
    \begin{align*}
        \frac{1}{2y(i)} &\geq 2(1-y(i)),\\
        1 &\geq 4y(i) (1-y(i)),
    \end{align*}
    where the last inequality follows again because its right hand side is maximized for $y(i) = \frac{1}{2}$. This maximal value is again equal to one and we have shown inequality~\eqref{ineq:v2>v1}.    It follows that $v_2(i) \geq v_1(i)=\log(2t(i))$ because
    \begin{equation*}
        \log(2(1-y(i)) \geq \log(2(1-x(i)-y(i))) = v_1(i)  .
    \end{equation*}
\end{proof}

\begin{proof}[Proof of Proposition~\ref{prop::properties of functions}] \ref{it::is central then}: Since $T(i)$ is central, we have $y(i)\le \frac{1}{2}$. Thus,
	$$
	u_2(i)=-\log (4x(i)y(i)) \ge -\log (4y^2(i)) =2 v_2(i).
	$$
	\ref{it:: is central and balanced}: Since $T(i)$ is central and balanced, we have $y(i)\le \frac{1}{2}$ and $y(i)/x(i)\le \beta / \alpha$. Hence,
	\begin{align*}
	u_2(i) &= - \log (4x(i)y(i)) \\
	&= - \log (4y^2(i)x(i)/y(i))\\
	&= - \log (4y^2(i))
 + \log(y(i)/x(i)) 	\\
 &\le 2v_2(i) + \log(\beta/\alpha).
\end{align*}
\ref{it:: cross, then balanced}: Without loss of generality, assume that $l(i-1) \le \frac{1}{2} \le r(i-1)$ and $l(i) > r(i)$. From Proposition~\ref{prop::trivial stuff}.\ref{it::cross implies correct transmission} it follows that $l(i)=\beta l(i-1)$ and $r(i)\ge \alpha r(i-1)$. Then 
$$
\frac{y(i)}{x(i)} = \frac{l(i)}{r(i)}  \le \frac{\beta l(i-1)}{\alpha r(i-1)} \le \frac{\beta r(i-1)}{\alpha r(i-1)} = \frac{\beta}{\alpha}.
$$
This means that $T(i)$ is balanced.
Similarly, we check that $T(i-1)$ is balanced
$$
\frac{y(i-1)}{x(i-1)} = \frac{r(i-1)}{l(i-1)} \le \frac{r(i)/\alpha}{l(i)/\beta} < \frac{\beta l(i) }{\alpha l(i)}=\frac{\beta}{\alpha}.
$$
\ref{it:: u1 is small and was central and correct transmission, then cross}: First we note that $u_1'<u_1''$. Without loss of generality, assume that $l(i-1) < r(i-1)\le \frac{1}{2}$. 
From the condition $u_1(i-1)\le u_1''$ it follows that
 $$
\frac{t(i-1)}{x(i-1)}= \frac{1-l(i-1)-r(i-1)}{l(i-1)}= \frac{1-r(i-1)}{l(i-1)} -1\le \frac{\beta}{\alpha}-1.
 $$
This implies $\beta l(i-1) +\alpha r(i-1) \ge \alpha$ or $\beta l(i-1)\ge \alpha - \alpha r(i-1)$. Since the transmission is correct, we have $l(i) = \beta l(i-1)$. Thus, $l(i)\ge \alpha - \alpha r(i-1)$. Since $r(i-1)\le \frac{1}{2}$, we have $l(i)> \alpha /2 > \alpha r(i-1) = r(i)$. This proves that there is a crossing between moment $i-1$ and moment $i$.\\
\ref{it:: u1 is large and cross, then is central}: Without loss of generality, we assume that $l(i-1)\le r(i-1)$. From the condition $u_1(i-1)\ge u_1'$, we derive
$$
\frac{1-2l(i-1)}{l(i-1)}\ge \frac{1-l(i-1)-r(i-1)}{l(i-1)}=\frac{t(i-1)}{x(i-1)}\ge 2^{u_1'}=\beta-\alpha=2\beta-2
$$
and $l(i-1)\le 1/(2\beta)$. Since there is a crossing, we have $\beta l(i-1)=l(i)\ge r(i)$ and hence $y(i) = l(i)\le \frac{1}{2}$.
\\
\ref{it:: u1 is large and cross, then was central}: Without loss of generality, assume that $l(i-1) < r(i-1)$ and $l(i)\ge r(i)$. First we note that a crossing may happen only when the transmission is correct. From condition $u_1(i-1)\ge u_1''$ we derive that
$$
\frac{t(i-1)}{x(i-1)}= \frac{1-l(i-1)-r(i-1)}{l(i-1)}= \frac{1-r(i-1)}{l(i-1)} -1\ge \frac{\beta}{\alpha}-1.
$$
This yields that $\beta l(i-1) +\alpha r(i-1) \le \alpha$. Since the transmission is correct, we have $l(i)=\beta l(i-1)$, and, thus, $l(i) + \alpha r(i-1) \le \alpha $ or $r(i-1)\le 1- l(i)/\alpha$. Because of the property $l(i)\ge r(i)$, we get $r(i-1)\le 1-r(i)/\alpha$. Since $r(i)\ge r(i-1)\alpha$, we obtain $r(i-1)\le 1 - r(i-1)$, which is equivalent to that $T(i-1)$ is central. To show that $T(i)$ is central, observe that $u_1(i-1)\ge u_1''\ge u_1'$ and we can use Proposition~\ref{prop::properties of functions}.\ref{it:: u1 is large and cross, then is central}.
\\
\ref{it:: u2 is large and cross, then is central}:  Suppose there is either a crossing between moment $i-1$ and moment $i$ or  a crossing between moment $i$ and moment $i+1$. Without loss of generality, assume that $l(i)\ge r(i)$ and either $l(i-1) < r(i-1)$ or $l(i+1) < r(i+1)$ .  Toward a contradiction, assume that $T(i)$ is not central and, thus, $y(i)=l(i)> \frac{1}{2}$. The condition $u_2(i)\ge u_2'$ is equivalent to that
$$
\frac{1-y(i)}{x(i)}=\frac{1-l(i)}{r(i)}\ge \frac{\beta}{\alpha}.
$$
This implies $\alpha l(i)+\beta r(i)\le \alpha$. Since $l(i)> \frac{1}{2}$, we have either $r(i)\ge \alpha r(i-1)>\alpha l(i-1) = \alpha l(i)/\beta > \alpha/(2\beta)$ or $\beta r(i) = r(i+1) > l(i+1) \ge l(i) \alpha > \alpha/2$. Thus, we obtain $\alpha l(i) + \beta r(i) > \alpha /2 + \alpha /2 = \alpha$.  This contradicts our assumption and, hence, $T(i)$ is central.
\end{proof}

\begin{proof}[Proof of Proposition~\ref{prop::u1increasing}]
    We consider first the part of the proposition where the condition $u_1(i-1) < u_1'$ holds. There are several cases that need to be distinguished here. We first consider the case when an error occurred. Then $t(i) \ge  \alpha t(i-1)$ (Proposition~\ref{prop::trivial stuff}.\ref{it::limits for t(i)}) and $x(i)=\alpha x(i-1)$ and thus, $u_1(i) = \log (t(i)/x(i)) \ge \log (t(i-1)/x(i-1))=u_1(i-1)$. Next we proceed to the case that the transmission was correct. If there is no crossing, then by Proposition \ref{prop::properties of functions}.\ref{it:: u1 is small and was central and correct transmission, then cross}  we have that $T(i-1)$ has to be non-central and, thus, $t(i)=\beta t(i-1)$ and $x(i)=\beta x(i-1)$, which implies $u_1(i)=u_1(i-1)$.
    
    Now suppose that there is a crossing which implies that no error occurred.
    We consider the case that $T(i-1)$ is not central. Because we have a crossing we know that $x(i) \le \beta x(i-1)$ and therefore
    \begin{equation*}
        2^{u_1(i)} = \frac{t(i)}{x(i)} \geq \frac{\beta t(i-1)}{\beta x(i-1)} = \frac{t(i-1)}{x(i-1)} = 2^{u_1(i-1)}  .
    \end{equation*}
    It remains to check the case when $T(i-1)$ is central. By the statement we need to show
    \begin{equation*}
        u_1(i) - u_1(i-1) = \log\left(\frac{1-x(i-1)\beta - y(i-1)\alpha}{y(i-1)\alpha}\right) - \log\left(\frac{1-x(i-1)-y(i-1)}{x(i-1)}\right) \geq \log \beta.
    \end{equation*}
    This is equivalent to
    \begin{align*}
        (1-x(i-1)\beta - y(i-1)\alpha) x(i-1) &\geq (1-x(i-1)-y(i-1))y(i-1) \alpha \beta ,\\
        x(i-1)\frac{\beta-\alpha}{2} (y(i-1)-x(i-1)) &\geq (1-x(i-1)-y(i-1)(y(i-1)\alpha \beta - x(i-1)).
    \end{align*}
    As the left hand side of this equation is always positive, we only have to consider the case where the right hand side has this property as well.
    Since $u_1(i-1)<u_1'$ we need prove that
    \begin{align*}
        y(i-1)-x(i-1) &> 2(y(i-1)\alpha \beta - x(i-1)),\\
        x(i-1) &> y(i-1)(2 \alpha \beta - 1).
    \end{align*}
   Because there is a crossing we know from Proposition~\ref{prop::properties of functions}.\ref{it:: cross, then balanced} that
    \begin{equation*}
        x(i-1)>y(i-1)\frac{\alpha}{\beta}  .
    \end{equation*}
    
To show $2 \alpha \beta - 1 < \frac{\alpha}{\beta}$ we use the chain of equivalent inequalities
    \begin{align*}
        2 \alpha \beta - 1 &< \frac{\alpha}{\beta},\\
        2 \alpha \beta^2 - \beta &< \alpha = 2 - \beta,\\
        \alpha \beta^2 < 1,
    \end{align*}
    where the last line follows from~\eqref{eq::key inequality}.
    This completes the first part of the proposition.

    Now we assume that $u_1(i-1) \geq u_1'$ and and show that in this case $u_1(i) \geq u_1'$. We assume that the transmission was incorrect.
    Then
    \begin{equation*}
        \frac{t(i)}{x(i)} \geq \frac{\alpha t(i-1)}{\alpha x(i-1)} = \frac{t(i-1)}{x(i-1)}  .
    \end{equation*}
    
    Therefore, from now on we only consider successful transmissions.
    First, we have a look at the case that $T(i-1)$ is not central.
    Then
    \begin{equation*}
        \frac{t(i)}{x(i)} \geq \frac{t(i-1) \beta}{x(i-1) \beta} = \frac{t(i-1)}{x(i-1)},
    \end{equation*}
    where the inequality followed because $\beta x(i-1) \geq x(i)$ with equality if there is no crossing and a strict inequality if there is a crossing. Therefore, we have shown the assertion for both cases.
    
    Next we analyze the case where $T(i-1)$ is central and there is no crossing
    \begin{equation*}
        \frac{t(i)}{x(i)} \geq \frac{t(i-1)}{\beta x(i-1)} \geq \frac{\beta - \alpha}{\alpha \beta} > \beta - \alpha.
    \end{equation*}
    The first inequality follows because correct transmission can never reduce the length of the true segment. The second one follows because otherwise by Proposition~\ref{prop::properties of functions}.\ref{it:: u1 is small and was central and correct transmission, then cross} there would be a crossing and the last inequality follows because $\alpha \beta < 1$.
    
    Next we consider the case that $T(i-1)$ is central and there is a crossing.
    We observe that since there is a crossing $T(i-1)$ is balanced by Proposition~\ref{prop::properties of functions}.\ref{it:: cross, then balanced}. Therefore, $y(i-1)/x(i-1) < \beta / \alpha$ and
    
    \begin{equation}\label{eq:proof_prop_2_1}
        \frac{t(i)}{x(i)} = \frac{1-x(i-1)\beta - y(i-1)\alpha}{y(i-1)\alpha} = \frac{1-x(i-1)\beta}{y(i-1)\alpha} - 1 .
    \end{equation}
    Since $u_1(i-1)\geq u_1'$ we obtain
    \begin{align*}
        \frac{1-x(i-1)-y(i-1)}{x(i-1)} &\geq \beta - \alpha,\\
        1-x(i-1)-y(i-1) &\geq (\beta - \alpha) x(i-1),\\
        1-x(i-1) - \frac{\beta - \alpha}{2} x(i-1) - y(i-1) &\geq \frac{\beta - \alpha}{2}x(i-1),\\
        1-\beta x(i-1) - y(i-1) &\geq \frac{\beta - \alpha}{2} x(i-1).
    \end{align*}
    Therefore, we get
    \begin{equation}\label{eq:proof_prop_2_2}
        \frac{1-x(i-1)\beta}{y(i-1)\alpha} \geq \frac{\frac{\beta - \alpha}{2}x(i-1) + y(i-1)}{y(i-1)\alpha} = \frac{1}{\alpha} + \frac{x(i-1) \frac{\beta - \alpha}{2}}{y(i-1)\alpha}.
    \end{equation}
    By using the fact that $T(i-1)$ is balanced we get
    \begin{equation}\label{eq:proof_prop_2_3}
        \frac{1}{\alpha} + \frac{x(i-1) \frac{\beta - \alpha}{2}}{y(i-1)\alpha} > \frac{1}{\alpha} + \frac{\beta - \alpha}{2 \beta} = \frac{2\beta + \alpha \beta  - \alpha^2}{2\alpha \beta} = \frac{(\alpha + \beta) \beta + \alpha \beta + \alpha^2 - 2\alpha^2}{2\alpha \beta},
    \end{equation}
    which we simplify to
    \begin{align}\label{eq:proof_prop_2_4}
        \frac{(\alpha+\beta)^2 - 2\alpha^2}{2\alpha \beta} = \frac{2-\alpha^2}{\alpha \beta} = 1 + \frac{2 - 2 \alpha}{\alpha \beta} 
        = 1 + \frac{\beta - \alpha}{\alpha \beta} > \beta - \alpha + 1,
    \end{align}
    where the last inequality follows because $1>\alpha \beta$.
    Combining equation \eqref{eq:proof_prop_2_1} and the inequalities \eqref{eq:proof_prop_2_2}, \eqref{eq:proof_prop_2_3}, \eqref{eq:proof_prop_2_4} completes the proof of this part.

    Now we proof that if $u_1(i-1)\geq u_1''$, then $u_1(i)\geq u_1''$. Seeking for a contradiction, assume $u_1(i-1)\geq u_1''$ and $u_1(i)< u_1''$, i.e., $u_1$ is decreasing. The cases that $T(i-1)$ is not central or the transmission was incorrect can be handled in the same manner as in the previous part of the proof. Therefore, $T(i-1)$ is central and the transmission is correct.  The following inequalities are equivalent
    \begin{align*}
    u_1(i-1)&\geq u_1''=\log(\beta/\alpha-1) ,\\
    \frac{1-x(i-1)-y(i-1)}{x(i-1)}&\geq \beta/\alpha-1 ,\\
    \alpha &\geq \alpha y(i-1)+\beta x(i-1).
    \end{align*}
    In the same manner, we can obtain the equivalent relations
    \begin{align*}
    u_1(i)&< u_1'',  \\
    \alpha&< \alpha y(i)+\beta x(i).
    \end{align*}
    Since $T(i-1)$ is central and the transmission is correct, we conclude that $\alpha y(i-1)+\beta x(i-1)=x(i)+y(i)$. Using this equality results in a contradiction
    \begin{equation*}
        \alpha\geq \alpha y(i-1)+\beta x(i-1)=y(i)+x(i)\geq \alpha y(i)+\beta x(i)>\alpha.
    \end{equation*}
\end{proof}

\begin{proof}[Proof of Proposition~\ref{prop::u2increasing}]
    We are first considering the case that both $T(i-1)$ and $T(i)$ are central. Whether there occurs an error or not it holds that
    \begin{equation*}
        u_2(i) - u_2(i-1) = -\log (4 x(i-1)\beta y(i-1) \alpha) + \log(4 x(i-1) y(i-1)) = -\log(\alpha \beta) > 0  ,
    \end{equation*}
    since $\alpha \beta < 1$.
    
    Now we consider the case that $T(i-1)$ is central and $T(i)$ is not central. Due to Proposition~\ref{prop::properties of functions}.\ref{it:: u1 is large and cross, then is central} there is no crossing and we know because of Proposition~\ref{prop::trivial stuff}.\ref{it::central to non-central without cross implies error} that an error occurred. It follows that
    \begin{equation*}
        u_2(i) - u_2(i-1) = \log\left(\frac{1-y(i)}{x(i)}\right) + \log(4x(i-1)y(i-1)) = \log\left(\frac{4(1-y(i)) y(i)}{\alpha \beta}\right).
    \end{equation*}
    Because $T(i)$ is not central we have that $\frac{1}{2} \leq y(i) \leq \frac{\beta}{2}$ where $y(i) = \frac{\beta}{2}$ minimizes the last logarithm, leading to
    \begin{equation*}
        \log\left(\frac{4(1-y(i)) y(i)}{\alpha \beta}\right) \geq \log 1 = 0  .
    \end{equation*}

    Next we are considering the case that $T(i-1)$ is not central and there occurs a crossing. Without loss of generality we consider that $l(i-1)<r(i-1)$. Due to Proposition~\ref{prop::properties of functions}.\ref{it:: u1 is large and cross, then is central} we know that $T(i)$ is central and therefore we only have to consider this case. Then, it holds that
    \begin{align*}
        u_2(i) - u_2(i-1) &= -\log(4l(i)r(i)) - \log \left (\frac{1-r(i-1)}{l(i-1)} \right ) = -\log \left ( \frac{4(1-r(i-1))l(i)r(i)}{l(i-1)} \right) \\
        &= -\log(4\beta(1-r(i-1)) r(i)) .
    \end{align*}
    Next we use that $r(i) = 1-\beta(t(i-1) + l(i-1)) = \frac{\alpha}{2} + \beta \left(r(i-1) - \frac{1}{2}\right)$ and obtain
    \begin{equation*}
        -\log(4\beta(1-r(i-1)) r(i)) = -\log\left(4\beta(1-r(i-1)) \left(\frac{\alpha}{2} + \beta \left(r(i-1) - \frac{1}{2}\right)\right) \right).
    \end{equation*}
    By looking at the quadratic form (in $r(i-1)$) we find its maximum for $r(i-1) = \frac{3 \beta - \alpha}{4\beta}$ and 
    \begin{equation*}
       -\log\left(4\beta(1-r(i-1)) \left(\frac{\alpha}{2} + \beta \left(r(i-1) - \frac{1}{2}\right)\right) \right) \geq -\log\left((\beta + \alpha)\left(\frac{\alpha}{2} + \frac{\beta - \alpha}{4}\right)\right) = -\log 1 = 0  .
    \end{equation*}
    
    
     It remains to check the case when $T(i-1)$ is not central and there is no crossing. The true segment $T(i)$ can be either central, or non-central. Since $\frac{1}{4x(i)y(i)} \geq \frac{1-y(i)}{x(i)}$, it holds that
    \begin{equation*}
        u_2(i) - u_2(i-1) \geq \log \left (\frac{1-y(i)}{x(i)} \right )  - \log \left (\frac{1-y(i-1)}{x(i-1)}\right ) = 0.
    \end{equation*}
    The last equality holds since there is no crossing and, thus, we have
    $$
        1-y(i)  = \begin{cases}
        (1-y(i-1))\beta, \quad \text{for correct transmission},\\
        (1-y(i-1))\alpha, \quad \text{for incorrect transmission},
        \end{cases}
   $$
    and
    $$
    x(i) = \begin{cases}
        x(i-1)\beta, \quad \text{for correct transmission},\\
        x(i-1)\alpha, \quad \text{for incorrect transmission}.
        \end{cases}
    $$
\end{proof}

\begin{proof}[Proof of Proposition~\ref{prop::function v1}]
    If there is a crossing, then because the transmission has to be correct due to Proposition~\ref{prop::trivial stuff}.\ref{it::cross implies correct transmission}. In that case the true segment is enlarged by Proposition~\ref{prop::trivial stuff}.\ref{it::correct transmission increases true segment}  and the first statement follows.
    For the second statement we first consider the case that $T(i-1)$ is not central. Then it is clear that
    \begin{equation*}
        v_1(i) - v_1(i-1) = \log t(i) - \log t(i-1) = \begin{cases}
        \log \beta, \quad \text{if the transmission is correct},\\
        \log \alpha, \quad \text{if the transmission is incorrect}.
        \end{cases}
    \end{equation*}
    
    The other case to be considered is if $T(i-1)$ is central and there is no crossing. In that case the transmission has to be incorrect because otherwise a crossing would occur, according to Proposition~\ref{prop::properties of functions}.\ref{it:: u1 is small and was central and correct transmission, then cross}.
    We have to show that
    \begin{equation*}
        \log t(i)-\log t(i-1) \geq \log \alpha, 
    \end{equation*}
    which is always true due to~Proposition~\ref{prop::trivial stuff}.\ref{it::limits for t(i)}.
\end{proof}

\begin{proof}[Proof of Proposition~\ref{prop::function v2}]
    We first consider the case that there is a crossing and $u_1(i-1)\geq u_1''$. Due to Proposition~\ref{prop::properties of functions}.\ref{it:: u1 is large and cross, then was central} we know that $T(i-1)$ and $T(i)$ are both central.   It follows that
    \begin{equation*}
        v_2(i) - v_2(i-1) =  -\log(2y(i))+\log(2y(i-1))= -\log(2x(i-1) \beta) + \log(2y(i-1)) \geq -\log \beta  .
    \end{equation*}

    For the remainder of the proof we are considering cases without a crossing.
    We consider the case that $T(i-1)$ and $T(i)$ are central. We obtain that
    \begin{equation*}
        v_2(i) - v_2(i-1) = \log\left(\frac{y(i-1)}{y(i)}\right) = \begin{cases}
        -\log \alpha \geq \log \beta, \quad \text{if the transmission was correct},\\
        -\log \beta \geq \log \alpha, \quad \text{if the transmission was incorrect}.
        \end{cases}
    \end{equation*}
    Next we consider the case that $T(i-1)$ is central and $T(i)$ is not central. This can only occur in the incorrect case (Proposition~\ref{prop::trivial stuff}.\ref{it::central to non-central without cross implies error}) and then
    \begin{equation*}
        v_2(i) - v_2(i-1) = \log(2(1-y(i))) + \log(2y(i-1)) = \log(4(1-y(i))y(i)) - \log \beta,
    \end{equation*}
    but since $y(i-1)\leq \frac{1}{2}$ we know that $y(i) \leq \frac{\beta}{2}$ and we get
    \begin{equation*}
      v_2(i) - v_2(i-1)=  \log(4(1-y(i)) y(i)) - \log \beta \geq \log(\alpha \beta) - \log \beta = \log \alpha.
    \end{equation*}
    In the considered case, it also holds that $v_2(i-1) < \log \beta$ since the following statements are equivalent
    \begin{align*}
        v_2(i-1) &< \log \beta,\\
        -\log(2y(i-1)) &< \log \beta,\\
        \frac{1}{2} &< \beta y(i-1) = y(i),
    \end{align*}
    where the last inequality holds because $T(i)$ is not central.
    
    
    Now we consider the case that $T(i-1)$ and $T(i)$ are both not central and there is no crossing. In this case it holds that
    \begin{align*}
        v_2(i) - v_2(i-1) &= \log(2(1-y(i))) - \log(2(1-y(i-1))) = \log\left(\frac{t(i) + x(i)}{t(i-1) + x(i-1)}\right) \\
        &= \begin{cases}
            \log \beta, \quad \text{if the transmission was correct},\\
            \log \alpha, \quad \text{if the transmission was incorrect}. 
        \end{cases}
    \end{align*}
    We show that $v_2(i-1)<\log \beta$ if the transmission was incorrect. It holds that
    \begin{equation*}
        \log(2(1-y(i-1))) \leq -\log(2y(i-1))  .
    \end{equation*}
    It follows that $v_2(i-1)<\log \beta$ because the following statements are equivalent
    \begin{align*}
        -\log(2y(i-1)) &< \log\beta ,\\
        \frac{1}{2y(i-1)} &< \beta ,\\
        \frac{1}{2} &< \beta y(i-1) = y(i) ,
    \end{align*}
    where the last line is true because $T(i)$ is not central and  the transmission was incorrect.
    
    Finally we consider the case that $T(i-1)$ is not central but $T(i)$ is central. In this case, the transmission is correct since otherwise it holds $y(i)\ge y(i-1)\ge \frac{1}{2}$. Furthermore, recall that
    \begin{equation*}
        -\log(2y(i)) \geq \log(2(1-y(i))) .
    \end{equation*}
    Therefore, we obtain
    \begin{align*}
        v_2(i) - v_2(i-1) = -\log(2y(i)) - \log(2(1-y(i-1))) \geq \log\left(\frac{1-y(i)}{1-y(i-1)}\right)= \log \beta,
    \end{align*}
    which concludes the proof.
\end{proof}

\begin{proof}[Proof of Proposition~\ref{prop::u2 and g}]
    From Proposition~\ref{prop::properties of functions}.\ref{it:: is central and balanced} and the inequality~\eqref{eq::key inequality} we deduce that
    \begin{equation}\label{eq::u2 big implies v2 big}
        v_2(i_0)>\frac{u_2(i_0)}{2}-\frac{\log(\beta/\alpha)}{2}\ge -\frac{1}{2}\log(\alpha^4\beta^6)\geq-\frac{1}{2}\log\beta^{-2}=\log\beta.
    \end{equation}
    Since $v_2(i_0)=-\log (2y(i_0))>\log\beta$, we have $y(i_0)<1/(2\beta)$.    From $y(i_0)<1/(2\beta)$ we obtain that 
    \begin{equation*}
        u_1(i_0)=\log\left(\frac{1-y(i_0)-x(i_0)}{x(i_0)}\right)\geq\log \left(\frac{1-2y(i_0)}{y(i_0)}\right)> u_1'.
    \end{equation*}
    From Proposition~\ref{prop::u1increasing} it follows that 
    $u_1(i)\geq u_1'$ for all $i\geq i_0$. Proposition~\ref{prop::u2increasing} implies that $u_2(i)\geq c$ for all $i\geq i_0$. Using  arguments as in~\eqref{eq::u2 big implies v2 big}, we get $v_2(i)>\log\beta$ for all $i$ such that $T(i)$ is central and balanced. Note that since $c\geq u_2'$
    , the inequality $u_2(i)\geq u_2'$ holds.
    
    Say that there are crossings at the moments $j_1, \ldots, j_k$, $i_0\leq j_1<\ldots<j_k<i_1$. Partition $[i_0, i_1)$ into the union $[i_0, j_1)\cup [j_1, j_1+1)\cup[j_1+1, j_2]\cup \ldots \cup [j_k, i_1)$. By Proposition~\ref{prop::properties of functions}.\ref{it:: cross, then balanced} and~\ref{prop::properties of functions}.\ref{it:: u2 is large and cross, then is central} the true segment is central and balanced at all the endpoints of the segments. Thus, all conditions of Proposition~\ref{prop::u2 and g} are satisfied. Moreover, if we prove the inequality~\eqref{eq::g and u_2} for these segments, it would also be true for their union. Therefore, it is sufficient to consider only two cases: 
    \begin{enumerate}
        \item $i_1=i_0+1$ and there is a crossing.
        \item there is no crossing between $i_0$ and $i_1$.
    \end{enumerate}
    
    The first case is trivial. There is a crossing between $i_0$ and $i_1=i_0+1$, hence the transmission is correct by Proposition~\ref{prop::trivial stuff}.\ref{it::cross implies correct transmission}. In this case $g(i_0, i_1)=\log\beta$. From the other hand $u_2(i_1)-u_2(i_0)=-\log(\alpha\beta)$ by Proposition~\ref{prop::u2increasing}, which is at least $\log\beta$ because of the inequality~\eqref{eq::key inequality}. 
    
    Proceed to the second case, where there is no crossing between $i_0$ and $i_1$. The smallest part is always the same, therefore 
    \begin{equation*}
        g(i_0, i_1)=e(i_0, i_1) \log\alpha + f(i_0, i_1)\log\beta=\log\left(\frac{x(i_1)}{x(i_0)}\right).
    \end{equation*}
    Denote $v_2(i_1)-v_2(i_0)$ and $u_2(i_1)-u_2(i_0)$ as $\Delta v_2$ and $\Delta u_2$ correspondingly.
    Since $T(i_0)$ and $T(i_1)$ are central, we have
    \begin{align*}
        \Delta u_2&=-\log (x(i_1)y(i_1))+\log (x(i_0)y(i_0))\\&=
    -\log y(i_1)+\log y(i_0)-\log\left(\frac{x(i_1)}{x(i_0)}\right)\\&=
    \Delta v_2-g(i_0, i_1),
    \end{align*}
    which is equivalent to 
    \begin{equation*}
        g(i_0, i_1) = \Delta v_2-\Delta u_2.
    \end{equation*}
    Using Proposition~\ref{prop::properties of functions}.\ref{it::is central then} and~\ref{prop::properties of functions}.\ref{it:: is central and balanced} for moments $i_1$ and $i_0$ respectively, we obtain
    \begin{equation*}
        g(i_0, i_1) = \Delta v_2-\Delta u_2\leq\frac{\Delta u_2+\log\left(\frac{\beta}{\alpha}\right)}{2}-\Delta u_2=\frac{-\Delta u_2+\log\left(\frac{\beta}{\alpha}\right)}{2}.
    \end{equation*}
    Recall that $y(i_0)<1/(2\beta)$, which implies that $y(i_0+1)<\frac{1}{2}$, i.e., $T(i_0+1)$ is central. Using Proposition~\ref{prop::u2increasing}, we deduce
    \begin{equation*}
        \Delta u_2=u_2(i_1)-u_2(i_0+1) + u_2(i_0+1)-u_2(i_0)\geq -\log(\alpha\beta).
    \end{equation*}
    Using this inequality and the inequality~\eqref{eq::key inequality}, we have
    \begin{align*}
        g(i_0, i_1)\leq \frac{-\Delta u_2+\log\left(\frac{\beta}{\alpha}\right)}{2}
        \leq \frac{\log(\alpha\beta)+\log\left(\frac{\beta}{\alpha}\right)}{2}=\log\beta
        \leq -\log(\alpha\beta)
        \leq \Delta u_2.
    \end{align*}
\end{proof}

\begin{proof}[Proof of Proposition~\ref{prop::function u}]
    Recall that $u_2(i)\geq u_1(i)$ for any $x(i)$, $y(i)$ and $t(i)$ by Proposition~\ref{prop::trivial stuff}.\ref{it::u2>u1}.
    
    If $u_1(i-1)<u_1'$, then $u_1(i)\geq u_1(i - 1)$ by Proposition~\ref{prop::u1increasing}, therefore, $u(i)\geq u_1(i)\geq u_1(i - 1)=u(i-1)$. If $T(i-1)$ is central and there is a crossing, then $u_1(i)- u_1(i - 1)\geq\log\beta$ by Proposition~\ref{prop::u1increasing}, hence, $u(i)-u(i-1)\geq\log\beta$.
    
    If $u_1(i-1)\geq u_1'$, then $u_1(i)\geq u_1'$ by Proposition~\ref{prop::u1increasing} and $u_2(i)\geq u_2(i-1)$ by Proposition~\ref{prop::u2increasing}. It means that $u(i)=u_2(i)\geq u_2(i-1)=u(i-1)$. If $T(i-1)$ is central and there is a crossing, then $T(i)$ is also central by Proposition~\ref{prop::properties of functions}.\ref{it:: u1 is large and cross, then is central}. So, we can use Proposition~\ref{prop::u2increasing}, which gives us 
    $u_2(i)- u_2(i - 1)\geq-\log(\alpha\beta)$. Since $\alpha\beta^2\leq 1$(see~\eqref{eq::key inequality})
    , this expression is not less than $\log\beta$.

\end{proof}

\begin{proof}[Proof of Proposition~\ref{prop::function v}]
    Recall that $v_2(i)\geq v_1(i)$ for any $x(i)$, $y(i)$ and $t(i)$ by Proposition~\ref{prop::trivial stuff}.\ref{it::v2>v1}.
    
    Suppose $u_1(i-1)<u_1''$. If there is no crossing, then the statement is implied by Proposition~\ref{prop::function v1}. If there is a crossing, then Proposition~\ref{prop::function v1} gives us $v_1(i)>v_1(i-1)$,
    therefore, $v(i)\geq v_1(i)>v_1(i-1)=v(i-1)$, i.e., $v(i)-v(i-1)>0$.
    For central $T(i-1)$ this bound is good enough. It remains to consider the case when $T(i-1)$ is not central. From $u_1(i-1)<u_1''$ we have $v(i)-v(i-1)\ge v_1(i)-v(i-1)=v_1(i)-v_1(i-1)$. Since there is a crossing, it is clear that transmission is correct. Condition $T(i-1)$ is not central implies that $t(i)=\beta t(i-1)$ and $v(i)-v(i-1)\ge v_1(i)-v_1(i-1)=\log\beta$.

    Suppose $u_1(i-1)\geq u_1''$. In this case $u_1(i)\geq u_1''$ by Proposition~\ref{prop::u1increasing}. If there is no crossing, then proposition follows from Proposition~\ref{prop::function v2}. If there is a crossing, then $T(i-1)$ is central by Proposition~\ref{prop::properties of functions}.\ref{it:: u1 is large and cross, then was central}. In this case $v(i)-v(i-1)=v_2(i)-v_2(i-1)\geq -\log\beta$ by Proposition~\ref{prop::function v2}.

\end{proof}







\section{Conclusion}\label{sec::conclusion}
In this paper, we have introduced a new problem statement of transmitting information through the adversarial insertion-deletion channel with feedback. We have shown a reduction of this problem to a problem of communication over the adversarial substitution channel. Thereby, the maximal asymptotic rate of feedback codes for the adversarial insertion-deletion channel has been established. In particular, the rate is positive whenever the fraction of insertion and deletion errors inflicted by the channel is less than $\frac{1}{2}$. We also revisit Horstein's algorithm~\cite{horstein1963sequential} for the adversarial substitution channel with feedback and present a more elaborate version of Zigangirov's analysis~\cite{zigangirov1976number}.

We emphasize that all results discussed in the paper concern the binary channel. A natural question that arises is whether it is possible to extend the methodology for the $q$-ary case. The best currently known results for the adversarial $q$-ary substitution channel with feedback are discussed in~\cite{lebedev2016coding}. In particular, for the fraction of errors $0<\tau<1/q$, the maximal asymptotic rate is established only for a countable number of values for $\tau$. Moreover, we point out that it is hard to generalize Zigangirov's arguments to the $q$-ary case and, in particular, find appropriate analogues of functions $u$ and $v$.
\section{Acknowledgment}
The authors are thankful to Zilin Jiang for the fruitful discussion on Zigangirov's proof. 


\begin{thebibliography}{99}
 \bibitem{Levenshtein1966binary} \textit{Levenshtein V.I.} Binary codes capable of correcting deletions, insertions, and reversals // Dokl. Akad. Nauk SSSR. 1965. V. 163. \textnumero~4.  P. 845--848.

\bibitem{varshamov1965code} \textit{Varshamov R.R., Tenengolts G.M.} Codes which correct single asymmetric errors (in Russian) // Automatika i Telemkhanika. 1965. V. 26, \textnumero~2. P. 288--292.

 \bibitem{cheraghchi2020overview} \textit{Cheraghchi M., Ribeiro J.}  An overview of capacity results for synchronization channels //   IEEE Trans. Inform. Theory. 2020. V. 67. \textnumero~6. P. 3207--3232.

 \bibitem{schulman1999asymptotically} \textit{Schulman L.J., Zuckerman D.} Asymptotically good codes correcting insertions, deletions, and transpositions // IEEE Trans. Inform. Theory. 1999. V. 45. \textnumero~7. P. 2552--2557.

 \bibitem{bukh2016improved2} \textit{Bukh B., Guruswami V.} An improved bound on the fraction of correctable deletions // Proceedings of the Twenty-Seventh Annual ACM-SIAM Symposium on Discrete Algorithms. Arlington, Virginia, USA, January 10 -- 12, 2016. P. 1893--1901.

 \bibitem{bukh2016improved} \textit{Bukh B., Guruswami V., H{\aa}stad, J.} An improved bound on the fraction of correctable deletions // IEEE Trans. Inform. Theory. 2016. V. 63. \textnumero~1. P. 93--103.

 \bibitem{plotkin1960binary} \textit{Plotkin M.} Binary codes with specified minimum distance // IRE Trans. Inform. Theory. 1960. V. 6. \textnumero~4. P. 445--450.

 \bibitem{berlekamp1964block} \textit{Berlekamp E.R.} Block coding with noiseless feedback. PhD Thesis, Cambridge: Massachusetts Institute of Technology, 1964.

 \bibitem{zigangirov1976number} \textit{Zigangirov K.Sh.} On the number of correctable errors for transmission over a binary symmetrical channel with feedback // Problems Inform. Transmission. 1976. V. 12. \textnumero~2. P. 85--97.

 \bibitem{horstein1963sequential} \textit{Horstein M.} Sequential transmission using noiseless feedback // IEEE Trans. Inform. Theory. 1963. V. 9. \textnumero~3. P. 136--143.

 \bibitem{schalkwijk1971class} \textit{Schalkwijk J.} A class of simple and optimal strategies for block coding on the binary symmetric channel with noiseless feedback // IEEE Trans. Inform. Theory. 1971. V. 17. \textnumero~3. P. 283--287.

 \bibitem{lebedev2016coding} \textit{Lebedev V.S.} Coding with noiseless feedback // Problems Inform. Transmission. 2016. V. 52. \textnumero~2. P. 103--113.

\end{thebibliography}

\end{document}